\documentclass[manuscript]{acmart}
\AtBeginDocument{%
  }

\setcopyright{none}

\usepackage{tikz}
\usetikzlibrary{arrows, decorations, backgrounds, positioning, patterns, shapes,
                circuits, circuits.ee.IEC, decorations.pathreplacing, 3d, calc}
\tikzset{>=latex}

\usepackage{amsmath}
\usepackage{algorithm}
\usepackage{algorithmic}

\usepackage{wrapfig}
\usepackage{subcaption}

\begin{document}

\title{Faster Random Walk-based Capacitance Extraction with Generalized Antithetic Sampling}

\author{Periklis Liaskovitis}
\affiliation{%
	\institution{Synopsys}
	\city{Athens}
	\country{GR}
}
\email{periklis.liaskovitis@synopsys.com}

\author{Marios Visvardis}
\affiliation{%
	\institution{Synopsys}
	\city{Athens}
	\country{GR}
}
\email{marios.visvardis@synopsys.com}

\author{Efthymios Efstathiou}
\affiliation{%
	\institution{Synopsys}
	\city{Athens}
	\country{GR}
}
\email{themis.efstathiou@synopsys.com}


\begin{abstract}
Floating random walk-based capacitance extraction has emerged in recent years as a tried
and true approach for extracting parasitic capacitance in very large scale integrated circuits.
Being a Monte Carlo method, its performance is dependent on the variance of sampled quantities
and variance reduction methods are crucial for the challenges posed by ever denser process
technologies and layout-dependent effects. In this work, we present a novel, universal
variance reduction method for floating random walk-based capacitance extraction, which is conceptually simple,
highly efficient and provably reduces variance in all extractions, especially when layout-dependent effects
are present. It is complementary to existing mathematical formulations for variance reduction and its
performance gains are experienced on top of theirs. Numerical experiments demonstrate substantial
such gains of up to 50\% in number of walks necessary as well as in actual extraction times
compared to the best previously proposed variance reduction approaches for the floating random-walk.
\end{abstract}




\maketitle

\section{Introduction}
Contemporary VLSI designs pose challenges for parasitic capacitance extraction
due to the sheer number of conductor metals, and dielectric layers involved.
The former can be in the vicinity of hundreds of thousands or even millions,
while the latter can be especially numerous and very thin \cite{Yu2013}.
Advanced technology processes exacerbate design size by imposing additional
complexity, for example in the form of non-stratified and conformal dielectrics
and other Layout-Dependent Effects (LDE), which themselves heavily influence
capacitance extraction \cite{Visvardis2023}.

Although classical field solver methods for the capacitance problem do exist, e.g.,
the Finite Difference (FDM), Finite and Boundary Element Methods (FEM-BEM),
the Floating Random Walk (FRW) Monte Carlo solver, offers very attractive advantages comparatively,
such as: robustness to geometric complexity thanks to lack of meshing,
output locality, i.e., allowing targeted solution on only a subset of domain points,
user-tunable accuracy and inherent scalability.
Thus, it has been steadily gaining popularity in recent years.

The FRW solver involves an application of Monte Carlo integration on a closed surface around the
conductor metal in question \cite{Yu2013, Yu2013a}.
Multiple integrals are estimated concurrently, one for each pairwise capacitance.
For consistent Monte Carlo estimation, i.e., given that the Monte Carlo means converge to the true values
of the constituent integrals, estimation error is only due to variance. For any finite number of samples
this is given by the ratio of the underlying random variable variance over the number of samples.
Thus, for a certain target accuracy, in order to expedite extraction, there exist two avenues:
either obtain the same number of samples faster, or reduce the inherent variance, so that
accuracy is increased with the same number of samples, hence variance reduction.

The specific application of Monte Carlo integration to capacitance extraction amounts to
obtaining samples through random walks starting from the integration surface.
However, the particular structure of the integral to be estimated
poses constraints to exactly what variance reduction methods can be applied, while application itself
is also not straightforward. For example, importance sampling can be applied both to choosing points
on the integration surface, but also to choosing points on the surface of the first transition domain.
The latter can be performed with different importance sampling functions, as showcased by
\cite{Batterywala2005} and \cite{Yu2013a}, potentially leading to huge differences in performance benefits.
The above discussion shows that variance reduction in the FRW solver requires careful consideration.
Ideally, any new variance reduction method should, as much as possible, complement
variance reduction mechanisms previously established in the state-of-the-art,
so as to reap their benefits already as a baseline.

In this work, we propose such a novel variance reduction approach for FRW-based capacitance extraction,
which effectively builds upon previous methods in the literature, but achieves notable further performance gains.
Our proposed approach utilizes sampling multiple points on the surface
of each first transition domain as starting points of random walks, as does \cite{Huang2024},
however, crucially, sampled points are \textit{guaranteed} to have weight values with signs opposite to each other.
The method to accomplish this does not strive to fulfill geometric constraints as previous work has attempted,
but is inherently data-driven, and relies only on the underlying Green's function gradient data.
Contributions of the paper are as follows:
\begin{itemize}
    \item A new algorithm to select starting points for walks on the surface of each first transition domain. The algorithm is impervious to the type of dielectrics (e.g., stratified or not) contained within the domain.
    \item Proof that the algorithm reduces variance of the Monte Carlo capacitance estimator, not just in expectation, but for every instantiation of the underlying stochastic weights, under very mild assumptions.
\end{itemize}

The rest of the paper is organized as follows:
In section \ref{sec:related-work}, prevalent variance reduction approaches for the FRW are discussed.
In section \ref{sec:preliminaries}, a basic overview of current random-walk-based capacitance extraction is given.
In section \ref{sec:gas}, the new algorithm is described and variance analysis is carried out.
Numerical results are presented in section \ref{sec:results}.
Finally, conclusions are drawn and relevant proofs are laid out in the appendix.

\section{Related Work} \label{sec:related-work}
Variance reduction specific to random walk for capacitance extraction has been an extensively studied theme in the literature
in recent years. After preliminary attempts \cite{Batterywala2005}, seminal work \cite{Yu2013, Yu2013a},
showed that important performance improvements could be obtained by:
i) importance sampling on the surface of the first transition domains,
and ii) stratified sampling, regarding the planar faces of the Gaussian surface as strata.
These works considered only simple rectilinear conductors.
Chains of circuit-connected conductors, "nets", complicated things however by requiring net-wide extraction,
for which a \textit{virtual} integration surface was introduced in \cite{Zhang2014}, the Virtual Gaussian Surface (VGS),
in place of the Block Gaussian Surface of \cite{Yu2013a}.

More recently, authors in \cite{Huang2024} present an approach that enables random walks to share
first transition domains, i.e., multiple walks can emanate from a single first transition domain.
The very idea of shooting multiple walks from a single transition domain is a central concept in itself
and is discussed extensively in \cite{Huang2024}, reaching the important conclusion that it is multiple walks shot from the
\textit{first} and not so much the subsequent transition domains that contributes substantially to variance reduction.
A theoretical variance reduction result is proved, however, the proof requires first transition
domains at best with stratified dielectrics and at worst with a single homogeneous dielectric,
assumptions, which are not very realistic in practice, \textit{especially} in the presence of LDE
\cite{Zhang2016, Visvardis2023}. A specific "exploration" strategy also has to be employed
to determine whether application of the proposed method in any given design will be beneficial or not, since,
as shown by their experiments, performance improvements are not always guaranteed with respect to current state-of-the-art.

This approach is the most conceptually similar to ours in the sense that we also propose multiple walks
from a single first transition domain. Nonetheless, the way we choose and manipulate starting points of walks on the surface
of the first transition domain is quite distinct, both as mathematical formulation and as practical application.
The method of \cite{Huang2024} chooses points on the first transition domain symmetrically
across the boundary of the Gaussian surface, changing the probability density function depending on the number of samples,
whereas we choose based on the sign of the weight values corresponding to the points, without requiring geometric symmetry
and keeping the probability density function the same for all points of the same transition domain.
This is a crucial difference, enabling our approach to offer practical variance reduction guarantees
for all designs, regardless of the arrangement of dielectrics contained within the first transition domains,
something that eludes \cite{Huang2024}. Additionally, because our method is universally better compared
to current state-of-the-art, there is no need to decide upon its usage or not in any given extraction.

\section{Preliminaries} \label{sec:preliminaries}
The capacitance extraction problem amounts to calculating all pair-wise capacitances between conductor
i, and all other conductors j, all kept at known electric potentials (the mathematical formulation of the problem
is equivalent either with conductor i at unit potential and all j, $j\neq{i}$ at zero potentials, or vice versa).
With fixed potentials, capacitance is directly related to charge induced on the surface, e.g., of conductor i
and according to Gauss's law this is given by the integral:

\begin{equation} \label{eq:1}
    Q = -\oint_G D(\mathbf{r}) \,d\mathbf{r} = -\oint_G \varepsilon(\mathbf{r})\nabla\phi(\mathbf{r}) \cdot \vv{\mathbf{n}}(\mathbf{r})\,d\mathbf{r}
\end{equation}

\noindent where $G$ is a closed integration surface around conductor $i$, and encompasses only conductor $i$,
$\varepsilon(\mathbf{r})$ is the dielectric permittivity function and $\phi(\mathbf{r})$ is the electric potential,
both defined at point $\mathbf{r}$ on the integration surface. This is an electrostatic setup governed by a differential equation of elliptic type, for which ample details can be found in
the literature \cite{Lecoz1992, Yu2013}. The random walk method employs Monte Carlo integration to compute the above integral,
effectively estimating its value via an ensemble average of samples of the integrand at many different points $\mathbf{r}$.

\begin{figure}[h!]
    \centering
    \def\scalefactor{0.6}
    \begin{tikzpicture}
    \def\s{\scalefactor}
    \filldraw[dotted, ultra thick, draw=red, fill=black!5] (-1*\s, -1*\s) rectangle (7*\s, 3*\s);
    \draw[->, semithick] (-0.5*\s, -2.2*\s) to[out=160, in=-80] (-0.8*\s, -1*\s);
    \node[anchor=west] at (-0.5*\s, -2.2*\s) {\small Gaussian surface enclosing conductor i};

    \filldraw[fill=black!15, draw=black] (0*\s, 0*\s) rectangle (6*\s, 2*\s);
    \filldraw[fill=black!15, draw=black] (5*\s, 6*\s) rectangle (11*\s, 8*\s);
    \node at (3*\s, 1*\s) {\small Conductor $i$};
    \node at (8*\s, 7*\s) {\small Conductor $j$};

    \draw[line width=0.5mm] (2*\s, 3*\s) circle[radius=2pt];
    \draw[line width=0.5mm] (7*\s, 0.5*\s) circle[radius=2pt];

    \foreach \rcx/\rcy/\rtx/\rty/\ra/\color in {2/3/2.7/4/1/blue!40!green,
                                           2.7/4/4.7/4.5/2/blue!40!green,
                                           4.7/4.5/5.5/6/1.5/blue!40!green,
                                           7/0.5/8/0.2/1/red!40!blue,
                                           8/0.2/9/2.2/2/red!40!blue,
                                           9/2.2/8.5/5.2/3/red!40!blue,
                                           8.5/5.2/8.7/6/0.8/red!40!blue} {
        \def\cx{\rcx*\s}
        \def\cy{\rcy*\s}
        \def\tx{\rtx*\s}
        \def\ty{\rty*\s}
        \def\a{\ra*\s}
        \draw[draw=\color, line width=0.2mm] (\cx - \a, \cy - \a) rectangle (\cx + \a, \cy + \a);
        \draw[->, line width=0.4mm, draw=\color] (\cx, \cy) to (\tx, \ty);
        \fill (\tx, \ty) circle[radius=2pt];
    }

    \draw[->, semithick] (0*\s, 6.2*\s) to[out=-90, in=180] (1*\s, 3.7*\s);
    \node[anchor=west] at (-1.6*\s, 6.5*\s) {\small First transition domain};
\end{tikzpicture}
    \caption{Floating random walk example. Two walks are launched from the integration surface around conductor $i$ and land on the surface of conductor $j$.}
    \label{fig:rw-typical-example}
\end{figure}

The integrand requires the value of the electrostatic potential, which is, however, unknown exactly on the surface $G$.
The only sets of points with known potentials are the conductor surfaces. The method takes advantage of a central property
of the associated elliptic operator, which is the existence of a propagation (or, equivalently, transition or generalized Poisson) kernel.
This is also known as \textit{surface Green's function} $P(\cdot)$ in the literature and enables the potential at any point $\mathbf{r}$
to be expressed as a weighted average of the potentials at points forming a closed surface enclosing $\mathbf{r}$:

\begin{equation} \label{eq:2}
    \phi(\mathbf{r}) = \oint_{S_1} P(\mathbf{r},\mathbf{r_1}) \phi(\mathbf{r_1}) \,d\mathbf{r_1}
\end{equation}

\noindent where $\mathbf{r_1}$ represents the position vector of a point on the enclosing surface, and is integrated over.

For capacitance extraction the most fitting such enclosure is a cube $S_1$ centered at $r$, which leads to the notion
of a transition cube, also referred to as transition domain \cite{Yu2013a}. This means that an unknown potential at a point $\mathbf{r}$
of $G$ can be expressed in terms of potentials at further away points $\mathbf{r_1}$ lying on the surface of transition domain $S_1$ centered
at $\mathbf{r}$. And, unknown potentials at $\mathbf{r_1}$ can, in turn, be expressed in terms of potentials at even further away points $\mathbf{r_2}$
on the surface of transition domain $S_2$ centered at $\mathbf{r_1}$.
Each of these integrals can be estimated in practice by a single sample estimator, i.e., utilizing a single sample on the surface of each respective transition domain
according to the surface Green's function.
This procedure of sampling on the surface of each consecutive transition domain can be repeated recursively, for as many steps $n$ as necessary,
until the point $\mathbf{r_n}$ lies on the surface of a conductor.
This recursion, effectively a Markov chain, forms the essence of the FRW method.
Substituting (\ref{eq:2}) into (\ref{eq:1}), for example once, one can obtain a two-step recursion as follows:

\begin{equation} \label{eq:3}
    Q = -\oint_G \varepsilon(\mathbf{r}) \oint_{S_1} \nabla_\mathbf{r} P(\mathbf{r},\mathbf{r_1}) \cdot \vv{\mathbf{n}}(\mathbf{r}) \oint_{S_2} P(\mathbf{r_1},\mathbf{r_2})\,d\mathbf{r_2}d\mathbf{r_1}d\mathbf{r}
\end{equation}

An illustration of how the method typically works is shown in Fig. \ref{fig:rw-typical-example}. Again, as in (\ref{eq:1}), (\ref{eq:3})
is an integral tackled by Monte Carlo integration, specifically by sampling (Markov) chains of transition domains starting on $G$,
creating transition domains $S_1$, ... , $S_n$ and marking what conductor the final transition domains end up at. Note that in (\ref{eq:3})
the inner integral over the first transition domain $S_1$ is quite distinct from integrals over all
subsequent transition domains $S_2\,...\,S_n$, in that it contains the gradient of the surface Green's function and therefore requires special treatment.
Variance reduction methods have been proposed to accelerate handling of the first transition domain and are primarily described in \cite{Yu2013, Yu2013a}.
In summary, importance sampling is applied on both $G$ and $S_1$, and the final form of the integral is:

\begin{equation} \label{eq:4}
    Q = \oint_G \frac{\varepsilon(\mathbf{r})}{F} \oint_{S_1} F q(\mathbf{r},\mathbf{r_1}) w(\mathbf{r},\mathbf{r_1}) \phi(\mathbf{r_1}) d\mathbf{r_1}d\mathbf{r}
\end{equation}

\noindent where the point $\mathbf{r_1}$ is sampled according to $q(\mathbf{r},\mathbf{r_1})$ and a weight value $w(\mathbf{r},\mathbf{r_1})$ is assigned to it:

\begin{equation} \label{eq:5}
    \begin{split}
    q(\mathbf{r},\mathbf{r_1}) &= \frac{\lvert \nabla_\mathbf{r} P(\mathbf{r},\mathbf{r_1}) \cdot \vv{\mathbf{n}}(\mathbf{r}) \rvert}{K(\mathbf{r})} \\
    w(\mathbf{r},\mathbf{r_1}) &=-\frac{K(\mathbf{r})}{L} \frac{\nabla_\mathbf{r} P(\mathbf{r},\mathbf{r_1}) \cdot \vv{\mathbf{n}}(\mathbf{r})}{\lvert \nabla_\mathbf{r} P(\mathbf{r},\mathbf{r_1}) \cdot \vv{\mathbf{n}}(\mathbf{r}) \rvert}
    \end{split}
\end{equation}

\noindent where $L$ is the side length of the cubic transition domain, $F=\oint_G \varepsilon(\mathbf{r}) d\mathbf{r}$, $K(\mathbf{r})=\oint_{S_1} \lvert \nabla_\mathbf{r} P(\mathbf{r},\mathbf{r_1}) \cdot \vv{\mathbf{n}}(\mathbf{r}) \rvert d\mathbf{r_1}$.
What is particularly important for the discussion that follows is that, given the above state-of-the-art formulation, for any given first transition domain $S_1$ there are only two possible weight values that can be associated with it, namely:

\begin{equation} \label{eq:6}
    w(\mathbf{r},\mathbf{r_1})=
\begin{cases}
    -\frac{K(\mathbf{r})}{L},& \text{if } \nabla_\mathbf{r} P(\mathbf{r},\mathbf{r_1}) \cdot \vv{\mathbf{n}}(\mathbf{r}) > 0 \\
    \frac{K(\mathbf{r})}{L},& \text{if } \nabla_\mathbf{r} P(\mathbf{r},\mathbf{r_1}) \cdot \vv{\mathbf{n}}(\mathbf{r}) < 0
\end{cases}
\end{equation}

In summary, for any given first transition domain, a point is sampled on its surface according to (\ref{eq:5}).
Furthermore, according to (\ref{eq:6}), this point corresponds to a weight value with fixed magnitude $\frac{K(\mathbf{r})}{L}$.
Only the sign of the weight value can vary, provided the transition domain is kept fixed.
The sign of the weight value is determined by the exact position of the point on the surface of the domain according to (\ref{eq:6}).
The random walk method estimates the value of capacitance by averaging many (i.e., hundreds of thousands or even millions)
weights of the form (\ref{eq:6}) corresponding to different transition domains $S_1^{(i)}$, all formed with their centers on the integration surface $G$.

\section{Generalized Antithetic Sampling} \label{sec:gas}

\subsection{Motivation and High-level idea}

It has been previously established \cite{Huang2024} that it is beneficial for variance reduction of the FRW method to sample not just one,
but multiple points on the surface of each first transition domain and simultaneously launch random walks
from all of them. We observe that this is conceptually very related to the idea of antithetic random variables in the literature of Monte Carlo integration \cite{Owen2013, Robert2010}.
An antithetic variable (equivalently, sample) strives to be such that the integrand takes on a somehow "opposite" value compared to that
of the original sample, i.e., it takes on a high value when the original sample takes on a low value and vice-versa.

The most counterintuitive aspect of antithetic sampling is that samples involved in the Monte Carlo estimator of an integral
can no longer be viewed individually. They can only be viewed in pairs:
instead of $2n$ independent samples, $n$ independent \textit{pairs} of samples are now used.
There is no restriction on how to obtain the two samples of an antithetic pair, as long as each of them follows the same \textit{marginal} distribution.
The joint distribution of the antithetic samples within a pair may be arbitrary.
The antithetics Monte Carlo estimator is then based on expectation \textit{against the common marginal distribution},
applied to the pairs as independent units and \textit{not against the joint distribution of the pair}.
The above is an important and largely missed result, generally referred to as \textit{generalized antithetic variables} \cite{Hammersley1956, Owen2013}.

For example, for a given integrand $f(\cdot)$ when $D$ is the unit square $[0,1]^2$ equipped with the uniform distribution, the initial sample is $\mathbf{x}$, obtained from the whole support $D$,
and the antithetic sample is $\tilde{\mathbf{x}}=1-\mathbf{x}$ component-wise, so that for the marginal pdfs we have $p_{marg}(\mathbf{x})=p_{marg}(\mathbf{\tilde{x}})$.
The joint distribution of the initial and antithetic sample could be different from the original uniform distribution,
nonetheless the estimator can be easily seen to be unbiased with respect to the common marginal \cite{Owen2013}:

\noindent\begin{minipage}{.5\textwidth}
    \begin{equation} \label{eq:theory}
        \begin{split}
            \hat{\mu}_{anti} &= \frac{1}{2n}\sum\limits_{i=1}^{n}(f(\mathbf{x}) + f(\tilde{\mathbf{x}})) \\
            E[\hat{\mu}_{anti}] &= \frac{1}{2n}\sum\limits_{i=1}^{n}2E[f(\mathbf{x})] = E[f(\mathbf{x})] \\
                                &= \int_{D}^{} p_{marg}(\mathbf{x})f(\mathbf{x})d\mathbf{x}
        \end{split}
    \end{equation}
\end{minipage}
\hfill
\begin{minipage}{.4\textwidth}
    \includegraphics[width=\linewidth]{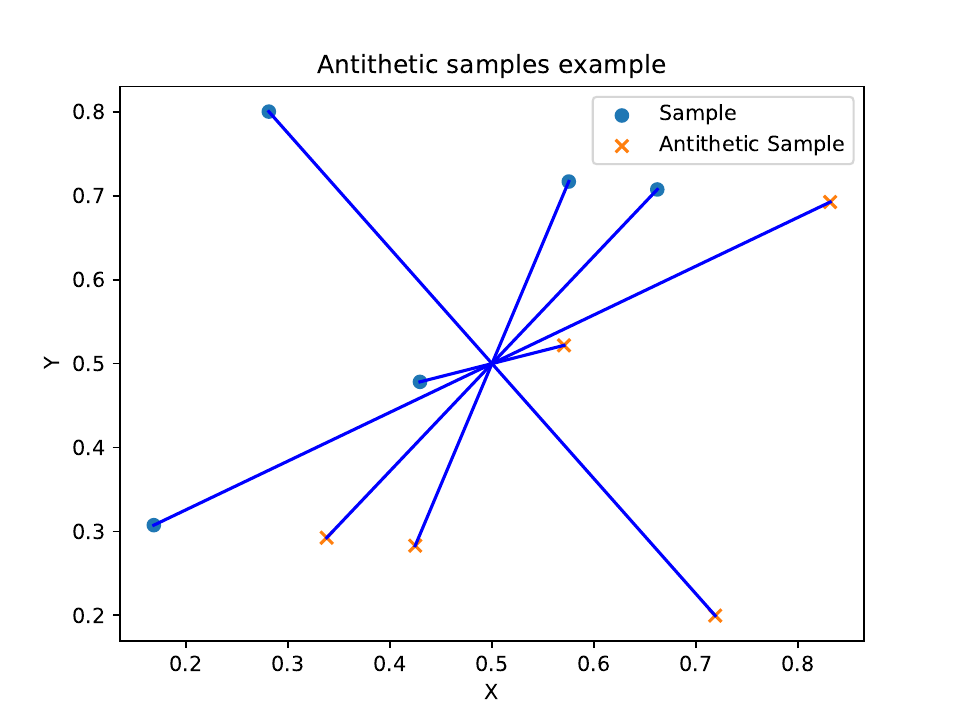}
    \captionof{figure}{Uniform antithetic samples on $[0,1]^2$}
\end{minipage}

The samples within each antithetic pair will not be independent,
however, for second-moment effects such as variance, only their correlation is of consequence.
In fact, optimal variance reduction is achieved when integrand values are \textit{maximally} negatively correlated \cite{Owen2013, Robert2010}.
Our main motivation in this work is that maximal variance reduction through antithetic variables
is not explicitly sought after or achieved in \cite{Huang2024}.
According to theory, maximal negative correlation of the weight values is necessary for optimality, and this amounts to every single weight value
obtained from a given first transition domain be paired with a weight value of opposite sign.
This is not at all guaranteed with the scheme of \cite{Huang2024}.

We describe a \textit{novel} method to obtain sample points on the surface of a given first transition domain, to obtain maximal negative correlation.
The first challenge towards this is to address the case of the relevant marginal distribution, i.e., the surface Green's function not being symmetric,
which is often the case in practice, e.g., in the presence of non-stratified dielectrics.
The relevant pdf in the integral is $q(\mathbf{r},\mathbf{r_1})$, its support set is the whole surface of the cubic transition domain
and the weight value $w(\mathbf{r},\mathbf{r_1})$ corresponding to each sample is the relevant integrand from (\ref{eq:4}).
To apply generalized antithetic sampling, we need both samples to be marginally distributed according to this exact same pdf.
Crucially, we want to deliberately arrange the samples, so that the second sample always corresponds to a weight value with sign opposite to
that of the first sample. By virtue of (\ref{eq:6}),
the magnitude of both of these weight values will be the same for any given first transition domain.

The pair of points on the surface of a given first transition domain with weight values of opposite signs is hereafter
referred to as an \textit{antithetic pair}.
Then, a way to achieve a single antithetic pair, while keeping the marginal pdf the same, is through sequential independent sampling from this pdf.
In other words, the same pdf is utilized repeatedly until the desired weight values of opposite signs
are first collected. An outline of the exact algorithm to achieve this is showcased as Algorithm \ref{alg:1}.

\begin{algorithm}[H]
    \caption{Algorithm to obtain $N$ generalized antithetic pairs of points on a given first transition domain} \label{alg:1}
    \begin{flushleft}
        \textbf{INPUT}: target number of antithetic pairs N, a transition domain and its associated $q(\mathbf{r},\mathbf{r_1}),K(\mathbf{r}), L$ \newline
        \textbf{OUTPUT}: $2N$ points in space to launch walks from
    \end{flushleft}
    \begin{algorithmic}[1]
        \FOR{1 up to N}
            \STATE Sample a point $Q$ on the surface of the first transition domain with $q(\mathbf{r},\mathbf{r_1})$;
            \STATE Store the point $Q$ and the corresponding weight value, according to (\ref{eq:6}), $w$;
            \REPEAT
                \STATE Sample a second point $\tilde{Q}$ on the surface of the first transition domain with $q(\mathbf{r},\mathbf{r_1})$;
                \STATE Store the point $\tilde{Q}$ and the corresponding weight value, according to (\ref{eq:6}), $\tilde{w}$;
            \UNTIL{$w \cdot \tilde{w} < 0$}
            \STATE Output the points \{$Q,\tilde{Q}$\} as a single antithetic pair;
        \ENDFOR
    \end{algorithmic}
\end{algorithm}

\noindent Algorithm \ref{alg:1} describes a discrete random process.
A single trial in the process is sampling a point on the surface of the first transition domain.
All trials in the process are obviously independent to each other and identically distributed (IID).
However, it is not immediately obvious why the marginal distributions of the initial and antithetic samples
actually coincide, and this is the subject of the next section.

\subsection{Algorithm Analysis}

A crucial property of the $m-th$ first transition domain $S_m$, with center $\mathbf{r_m}$ and Gaussian surface normal $\vv{\mathbf{n}}(\mathbf{r_m})$,
is the probability $p_m$ of obtaining a negative weight value by sampling on its surface, defined according to:

\begin{equation}
    \begin{split}
        S_m^{+} &= \{\mathbf{r_1^{(m)}}: \nabla_\mathbf{r} P(\mathbf{r_m},\mathbf{r_1^{(m)})} \cdot \vv{\mathbf{n}}(\mathbf{r_m}) > 0\} \\
        S_m^{-} &= \{\mathbf{r_1^{(m)}}: \nabla_\mathbf{r} P(\mathbf{r_m},\mathbf{r_1^{(m)})} \cdot \vv{\mathbf{n}}(\mathbf{r_m}) < 0\} \\
        p_m     &= \oint_{S_m^{+}} q(\mathbf{r_m},\mathbf{r_1^{(m)}}) d\mathbf{r_1} = \oint_{S_m^{+}} \frac{\nabla_\mathbf{r} P(\mathbf{r_m},\mathbf{r_1^{(m)}}) \cdot \vv{\mathbf{n}}(\mathbf{r_m})}{K(\mathbf{r_m})} d\mathbf{r_1} \label{eq:12}
    \end{split}
\end{equation}

\noindent where the pdf for sampling on the $m-th$ first transition domain $q(\mathbf{r_m},\mathbf{r_1^{(m)}})$ is given analytically by (\ref{eq:5}).
At first glance, the value of $p_m$ can seem to be different per first transition domain,
because the center of each such domain is stochastically chosen and may contain a slightly different configuration of dielectrics.
A crucial result, necessary for the analysis that follows, is that it is in fact \textit{constant}, regardless of transition domain:

\begin{lemma} \label{prop:l1}
When sampling on the surface of an arbitrary first transition domain $S_m$ with $q(\mathbf{r_m},\mathbf{r_1^{(m)}})$ it holds that $p_m = 1 - p_m = p = 0.5$.
\end{lemma}

\noindent The above is a known result from the literature of FRW, found in the appendix of \cite{Yu2013}.
It is based on the fact that the integrals over the sets $S_m^{+}$ and $S_m^{-}$ are of equal magnitudes and opposite signs,
by virtue of the surface Green's function being a proper probability density function at every point in space $\mathbf{r_m}$:

\begin{equation}
    \begin{split}
        K^{+}(\mathbf{r_m}) &= \oint_{S_m^{+}} \nabla_\mathbf{r} P(\mathbf{r_m},\mathbf{r_1^{(m)}}) \cdot \vv{\mathbf{n}}(\mathbf{r_m}) \,d\mathbf{r_1} \\
        K^{-}(\mathbf{r_m}) &= \oint_{S_m^{-}} \nabla_\mathbf{r} P(\mathbf{r_m},\mathbf{r_1^{(m)}}) \cdot \vv{\mathbf{n}}(\mathbf{r_m}) \,d\mathbf{r_1} \\
        K(\mathbf{r_m}) &= K^{+}(\mathbf{r_m}) + K^{-}(\mathbf{r_m}) = 2 \cdot K^{+}(\mathbf{r_m}) \implies p_m = \frac{K^{+}(\mathbf{r_m})}{2 \cdot K^{+}(\mathbf{r_m})} = \frac{1}{2}
    \end{split}
\end{equation}

We can now prove that Algorithm \ref{alg:1} fulfills the same marginals prerequisite of generalized antithetic sampling and can
lead to an unbiased estimator based on (\ref{eq:theory}) in the next section:

\begin{lemma} \label{prop:l2}
The initial and antithetic samples obtained by Algorithm \ref{alg:1} have the same marginal distribution.
\end{lemma}

\begin{proof}

Consider the random variables $\mathbf{Q_m}$ and $\mathbf{\tilde{Q}_m}$ on the surface of the transition domain,
corresponding to the initial and antithetic samples in Algorithm \ref{alg:1} respectively.
The marginal probability density of the initial variable $\mathbf{Q_m}$ , $q_{init}(\mathbf{r_m},\mathbf{r_1^{(m)}})$, is obviously equal to $q(\mathbf{r_m},\mathbf{r_1^{(m)}})$.
On the other hand, $\mathbf{\tilde{Q}_m}$ can be regarded as being jointly distributed with a binary random variable $Z_m$,
taking on two distinct values:

\begin{align*}
    z_m^{+} &:= \mathbf{Q_m} \in S_m^{+} & Prob(Z_m = z_m^{+}) &= p_m \\
    z_m^{-} &:= \mathbf{Q_m} \in S_m^{-} & Prob(Z_m = z_m^{-}) &= 1 - p_m
\end{align*}

\noindent Consider the probability measure induced by $q(\mathbf{r_m},\mathbf{r_1^{(m)}})$ as $\mu_{q_\mathbf{r_m}}(\cdot)$
(note that $Z_m$ induces a completely different measure).
If the value of $Z_m$ is known and fixed, e.g., $Z_m = z_m^{+}$, then, by Algorithm \ref{alg:1}, $\mathbf{\tilde{Q}_m} \in S_m^{-}$
and $\mathbf{\tilde{Q}_m}$ is distributed according only to $\mu_{q_\mathbf{r_m}}(\cdot)$.
Specifically, the probability of the event that $\mathbf{\tilde{Q}_m}$ lies within a neighborhood $B_\mathbf{r_1^{(m)}}$ of point $\mathbf{r_1^{(m)}}$,
assuming an appropriate notion of neighborhood and limit thereof on the surface of the transition domain,
will be the ratio of the $\mu_{q_\mathbf{r_m}}(\cdot)$ measures of $B_\mathbf{r_1^{(m)}}$ and the whole of $S_m^{-}$,
since $B_\mathbf{r_1^{(m)}}$ is then a subset of $S_m^{-}$:

\begin{equation}
    \begin{split}
        Prob(\mathbf{\tilde{Q}_m} \in B_\mathbf{r_1^{(m)}} | Z_m = z_m^{+}) &= Prob(\mathbf{\tilde{Q}_m} \in B_\mathbf{r_1^{(m)}} | \mathbf{\tilde{Q}_m} \in S_m^{-}) =
        \frac{\mu_{q_\mathbf{r_m}}B_\mathbf{r_1^{(m)}}}{\mu_{q_\mathbf{r_m}}(S_m^{-})}
        \implies \text{($\lim{B_\mathbf{r_1^{(m)}} \to 0}$)} \\
        \tilde{q}(\mathbf{r_m},\mathbf{r_1^{(m)}} | Z_m = z_m^{+}) &=
        \begin{cases}
           0 & \text{, } \mathbf{r_1^{(m)}} \in S_m^{+} \\
            \frac{q(\mathbf{r_m}, \mathbf{r_1^{(m)}})}{1 - p_m} & \text{, } \mathbf{r_1^{(m)}} \in S_m^{-} \\
        \end{cases} = \frac{q(\mathbf{r_m}, \mathbf{r_1^{(m)}})}{1 - p_m} \cdot \mathbf {1} _{S_m^{-}}(\mathbf{r_1^{(m)}}) \label{eq:14}
    \end{split}
\end{equation}

\noindent where $\tilde{q}(\mathbf{r_m},\mathbf{r_1^{(m)}} | Z_m = z_m^{+})$ is the corresponding conditional probability density function
and ${1} _{S_m^{-}}$ is the indicator function for the set $S_m^{-}$ on the surface of the transition domain.
It is easy to see that $\tilde{q}(\mathbf{r_m},\mathbf{r_1^{(m)}} | Z_m = z_m^{+})$ is properly normalized on its respective support set.
Similarly for conditioning upon $Z_m = z_m^{-}$:

\begin{equation}
    \begin{split}
        \tilde{q}(\mathbf{r_m},\mathbf{r_1^{(m)}} | Z_m = z_m^{-}) &=
        \begin{cases}
            \frac{q(\mathbf{r_m}, \mathbf{r_1^{(m)}})}{p_m} & \text{, } \mathbf{r_1^{(m)}} \in S_m^{+} \\
            0 & \text{, } \mathbf{r_1^{(m)}} \in S_m^{-}
        \end{cases} = \frac{q(\mathbf{r_m}, \mathbf{r_1^{(m)}})}{p_m} \cdot \mathbf {1} _{S_m^{+}}(\mathbf{r_1^{(m)}}) \label{eq:15}
    \end{split}
\end{equation}

\noindent Now we are in a position to marginalize the joint probability distribution of $(\mathbf{\tilde{Q}_m}, Z_m)$ over $Z_m$
utilizing the definition of conditional probability:

\begin{equation}
    \begin{split}
        Prob(\mathbf{\tilde{Q}_m} \in B_\mathbf{r_1^{(m)}}) &= \sum_{z_m=\{z_m^{-},z_m^{+}\}}Prob(\mathbf{\tilde{Q}_m} \in B_\mathbf{r_1^{(m)}}, Z_m = z_m) \\
                                                            &= \sum_{z_m=\{z_m^{-},z_m^{+}\}}Prob(\mathbf{\tilde{Q}_m} \in B_\mathbf{r_1^{(m)}} | Z_m = z_m) \cdot Prob(Z_m = z_m) \\
                                                            &= Prob(\mathbf{\tilde{Q}_m} \in B_\mathbf{r_1^{(m)}} | Z_m = z_m^{-}) \cdot Prob(Z_m = z_m^{-}) \\
                                                            &+ Prob(\mathbf{\tilde{Q}_m} \in B_\mathbf{r_1^{(m)}} | Z_m = z_m^{+}) \cdot Prob(Z_m = z_m^{+}) \\
                                                            &= Prob(\mathbf{\tilde{Q}_m} \in B_\mathbf{r_1^{(m)}} | Z_m = z_m^{-}) \cdot (1 - p_m) \\
                                                            &+ Prob(\mathbf{\tilde{Q}_m} \in B_\mathbf{r_1^{(m)}} | Z_m = z_m^{+}) \cdot p_m \implies \text{($\lim{B_\mathbf{r_1^{(m)}} \to 0}$ and (\ref{eq:14}), (\ref{eq:15}))} \\
                  q_{anti}(\mathbf{r_m},\mathbf{r_1^{(m)}}) &= \frac{1 - p_m}{p_m} \cdot q(\mathbf{r_m}, \mathbf{r_1^{(m)}}) \cdot \mathbf {1} _{S_m^{+}}(\mathbf{r_1^{(m)}}) + \frac{p_m}{1 - p_m} \cdot q(\mathbf{r_m}, \mathbf{r_1^{(m)}}) \cdot \mathbf {1} _{S_m^{-}}(\mathbf{r_1^{(m)}})
    \end{split}
\end{equation}

\noindent The analysis so far has been general, without having taken into account Lemma \ref{prop:l1}.
With $p_m = 1 - p_m = 0.5$, the last expression readily becomes equal to $q_{init}(\mathbf{r_m},\mathbf{r_1^{(m)}})$:

\begin{equation}
    q_{anti}(\mathbf{r_m},\mathbf{r_1^{(m)}}) = q(\mathbf{r_m}, \mathbf{r_1^{(m)}}) \cdot \mathbf {1} _{S_m^{+}}(\mathbf{r_1^{(m)}}) + q(\mathbf{r_m}, \mathbf{r_1^{(m)}}) \cdot \mathbf {1} _{S_m^{-}}(\mathbf{r_1^{(m)}}) = q(\mathbf{r_m}, \mathbf{r_1^{(m)}})
\end{equation}

\end{proof}

Regarding performance, Algorithm \ref{alg:1} can be represented by a Markov chain,
the average time to termination of which is shown in the appendix to also depend on the relative magnitudes of $p_m$ and $1 - p_m$.
With equal such probabilities, the chain terminates on average in the minimum possible of 3 steps, which makes it quite efficient, considering that the sampling pdf $q(\mathbf{r},\mathbf{r_1})$
used for the transitions is only ever computed once for every first transition domain.

A typical application of the FRW method employing Algorithm \ref{alg:1} is illustrated in Fig. \ref{fig:rw-antithetics-example}.
The points selected are agnostic to their geometric position on the surface of the transition domain, as opposed to the SMS scheme in \cite{Huang2024}.
It may well be the case that there is no geometric symmetry across the Gaussian surface for the two points sampled,
as is showcased in Fig. \ref{fig:rw-antithetics-example} for both initial transition domains depicted.
Only the sign of $\nabla_\mathbf{r} P(\mathbf{r},\mathbf{r_1}) \cdot \vv{\mathbf{n}}(\mathbf{r})$ actually matters.
The algorithm can be extended in a straightforward manner for the purpose of obtaining not just one,
but multiple antithetic pairs per first transition domain.
All antithetic pairs produced are also obviously IID as required.

\begin{figure}[h!]
    \centering
    \def\scalefactor{0.6}
    \begin{tikzpicture}
    \def\s{\scalefactor}
    \filldraw[dotted, ultra thick, draw=red, fill=black!5] (-1*\s, -1*\s) rectangle (7*\s, 3*\s);
    \draw[->, semithick] (-0.5*\s, -2.2*\s) to[out=160, in=-80] (-0.8*\s, -1*\s);
    \node[anchor=west] at (-0.5*\s, -2.2*\s) {\small Gaussian surface enclosing conductor i};

    \filldraw[fill=black!15, draw=black] (0*\s, 0*\s) rectangle (6*\s, 2*\s);
    \filldraw[fill=black!15, draw=black] (5*\s, 6*\s) rectangle (11*\s, 8*\s);
    \node at (3*\s, 1*\s) {\small Conductor $i$};
    \node at (8*\s, 7*\s) {\small Conductor $j$};

    \draw[line width=0.5mm] (2*\s, 3*\s) circle[radius=2pt];
    \draw[line width=0.5mm] (7*\s, 0.5*\s) circle[radius=2pt];

    \foreach \rcx/\rcy/\rtx/\rty/\ra/\color/\style in {2/3/2.7/4/1/blue!40!green/solid, 
                                                       2.7/4/4.7/4.5/2/blue!40!green/solid, 
                                                       4.7/4.5/5.5/6/1.5/blue!40!green/solid, 
                                                       2/3/1/3.5/1/blue!40!green/dashed, 
                                                       1/3.5/1.5/5/1.5/blue!40!green/dashed, 
                                                       1.5/5/4.5/7.8/3/blue!40!green/dashed, 
                                                       4.5/7.8/5/7.5/0.5/blue!40!green/dashed, 
                                                       7/0.5/8/0.2/1/red!40!blue/solid, 
                                                       8/0.2/9/2.2/2/red!40!blue/solid, 
                                                       9/2.2/8.5/5.2/3/red!40!blue/solid, 
                                                       8.5/5.2/8.7/6/0.8/red!40!blue/solid, 
                                                       7/0.5/6/1.2/1/red!40!blue/dashed} { 
        \def\cx{\rcx*\s}
        \def\cy{\rcy*\s}
        \def\tx{\rtx*\s}
        \def\ty{\rty*\s}
        \def\a{\ra*\s}
        \draw[draw=\color, line width=0.2mm, \style] (\cx - \a, \cy - \a) rectangle (\cx + \a, \cy + \a);
        \draw[->, line width=0.4mm, draw=\color, \style] (\cx, \cy) to (\tx, \ty);
        \fill (\tx, \ty) circle[radius=2pt];
    }

    \draw[->, semithick] (0*\s, 6.2*\s) to[out=-90, in=180] (1*\s, 3.6*\s);
    \draw[->, semithick] (0*\s, 6.2*\s) to[out=0, in=75] (2.7*\s, 4.0*\s);
    \node[anchor=west] at (-1.6*\s, 6.5*\s) {\small Antithetic samples};
\end{tikzpicture}
    \caption{Floating random walk with antithetic sampling example. Two walks are initiated from each first transition cube from antithetic samples on their surfaces.
             The walks starting from antithetic samples are drawn with dashed lines. For the first sample (green lines) both walks land on the surface of conductor $j$, while for the second (purple lines) one finishes on conductor $j$ and the other on conductor $i$.}
    \label{fig:rw-antithetics-example}
\end{figure}

\subsection{Variance Computation} \label{sec:estimator}

In summary, when employing a single antithetic pair, we sample two points on the surface of any given first transition domain according to Algorithm \ref{alg:1}
and launch walks from \textit{both} of these points. Each first transition domain centered on a point of the VGS thus contributes two full random walks
and two weight values to the final capacitance calculation.
This is similar in concept to \cite{Huang2024}, with the important difference that the two weight values are \textit{guaranteed} to be of opposite signs.
Although more than one antithetic pairs can be obtained on the surface of each first transition domain,
in what follows, for simplicity of exposition, we restrict our analysis to a single antithetic pair per first transition domain.

As also discussed in \cite{Huang2024}, a potential benefit of starting multiple walks from the same first transition domain, instead of
a single walk, is reduction in the number of samples on the VGS necessary for convergence of the FRW method.
Taking fewer samples on the VGS means forming fewer first transition domains centered on those samples,
and correspondingly fewer matches through a Green's function precomputation scheme.
Specifically, precomputation matching and calculation of $q(\mathbf{r},\mathbf{r_1})$ and $K(\mathbf{r})$ only need to take place
once for every $N$ random walks.

An even more important potential benefit is variance reduction achieved. In this section,
we specifically examine how the mean and variance of the Monte Carlo estimator for capacitance are affected by the
correlation among weight values introduced by generalized antithetic sampling (hereafter referred to as GAS).
Each transition domain $i$ on the VGS out of $n$ will contribute two full random walks and two weight values to the final
capacitance calculation to a total of $2n$.

To form the actual capacitance estimator we need to take into account the potential $\phi_{i}$ of the conductor
each walk eventually lands on. The convention we use, typical for capacitance extraction,
is set the potential of the master conductor to 0 V, and the potential of all the other conductors to 1 V, the master conductor being where walks start from.
With this convention, the estimator essentially comprises weights only from walks that do \textit{not} land on the master conductor.

\newcommand\numberthis{\addtocounter{equation}{1}\tag{\theequation}}

\begin{align*}
    w_{i} &= \frac{K_{i}}{2L_{i}} & X &= \frac{1}{n_{+}}\sum\limits_{i=1}^{n_{+}}w_{i} \cdot \phi_{i} \\
    \tilde{w}_{i} &= -\frac{K_{i}}{2L_{i}} & \tilde{X} &= \frac{1}{n_{-}}\sum\limits_{j=1}^{n_{-}}\tilde{w}_{j} \cdot \tilde{\phi}_{j} \\
    w_{i} \cdot \tilde{w}_{i} &= -w_{i}^2 \leq 0 & n_{+} &= n_{-} = n, \phi_{i}, \tilde{\phi}_{j} \in \{0, 1\} \numberthis \label{eq:7}
\end{align*}

\noindent The factor of 2 in the denominator of the weights above, serves to make the estimator unbiased in light of the fact that there are now $2n$ random walks or weight values,
from each $n$ transition domains on the VGS.
If all positive weight values from first transition domains $i$ are arranged in a positive group and all negative values in a corresponding negative group,
then $X$ and $\tilde{X}$ are the random variables representing the sample means of these two groups.
The equality of group sizes is evident since, for any first transition domain, one positive and one negative weight value will always be produced by construction,
see Algorithm \ref{alg:1}.

The proposed GAS estimator is simply the sum of the above group sample means:

\begin{equation}
    \hat{\mu}_{GAS} = F \cdot (X + \tilde{X}) = F \cdot \frac{1}{n}(\sum\limits_{i=1}^{n}w_i \phi_{i} + \sum\limits_{j=1}^{n}\tilde{w}_j \tilde{\phi}_{j}) \label{eq:8}
\end{equation}

\noindent where $F$ is the permittivity integral over the VGS included in (\ref{eq:5}).
GAS is similar to the the Importance Sampling-Stratified Sampling (IS+SS) estimator already used in FRW methodology \cite{Yu2013, Yu2013a}, where
$X$ and $\tilde{X}$ can best be thought of as the sample means per positive and negative stratum respectively.
The subtle difference is that GAS jointly considers $2n$ weight values for each $n$ first transition domains on the VGS,
whereas IS+SS considers only $n$ weight values for the same domains, approximately $n/2$ per stratum.
Also, unlike IS+SS, in GAS, $X$ and $\tilde{X}$ are not independent, therefore
we need to explicitly examine the mean and variance of the sum of $X$ and $\tilde{X}$ according to
the antithetics method \cite{Owen2013}:

\begin{equation}
    \begin{split}
        E[\hat{\mu}_{GAS}] &= F \cdot (E[X] + E[\tilde{X}]) \\
        Var[\hat{\mu}_{GAS}] &= F^2 \cdot (Var[X] + Var[\tilde{X}] + 2Cov[X, \tilde{X}]) \label{eq:9}
    \end{split}
\end{equation}

The expectation of the sum of group sample means is equal to the sum of the expectations of the sample means,
exactly as in the IS+SS case.
In the appendix, we prove that the expected value of the GAS estimator is indeed unbiased for the capacitance problem
within the framework of generalized antithetics.
By contrast, the variance of the sample means is the sum of the sample variances, modified by the \textit{covariance}
of the sample means across the two groups.

Combining with (\ref{eq:7}), the covariance of $X$ and $\tilde{X}$ can be related to the moments of $w_{i}$, $\tilde{w}_{j}$ as follows:

\begin{equation}
    \begin{split}
        Cov[X, \tilde{X}] &= E[X \cdot \tilde{X}] - E[X] \cdot E[\tilde{X}] = \frac{1}{n^2}\sum\limits_{i=1}^{n}E[w_{i}\tilde{w}_{i}]\phi_{i}\tilde{\phi}_{i} - \frac{1}{n^2}\sum\limits_{i=1}^{n}E[w_{i}]E[\tilde{w}_{i}]\phi_{i}\tilde{\phi}_{i} \label{eq:10}
    \end{split}
\end{equation}

\noindent The equalities have taken into account that $w_{i}$, $\tilde{w}_{j}$ are uncorrelated for ${i}\neq{j}$, hence $E[w_{i}\tilde{w}_{j}] = E[w_{i}]E[\tilde{w}_{j}]$.
This is covariance in expectation. Monte Carlo (or sample) covariance can be estimated from actual weight values $w$, $\tilde{w}$ from (\ref{eq:7}) and (\ref{eq:10})
by using sample means in place of actual expectation (as is routinely done for sample variance):

\begin{equation} \label{eq:11}
    Cov[X, \tilde{X}] \approx \Delta = \frac{1}{n^2}[\sum\limits_{i=1}^{n^\prime}w_{i}\tilde{w}_{i} - \frac{1}{n^\prime}(\sum\limits_{i=1}^{n^\prime}w_{i})(\sum\limits_{j=1}^{n^\prime}\tilde{w}_{j})]
\end{equation}

\noindent The above also takes into account that, for non-zero target potentials,
only $\phi_{i}\tilde{\phi}_{i}=1$ and $\phi_{i}\tilde{\phi}_{j}=1$ terms survive.
Effectively, correlation comes only from antithetic samples where both walks land on conductors other than the master conductor
and $n^\prime$ is the number of these walks, where $n^\prime<n$.
We can now prove the following result (full proof in the appendix):

\begin{theorem}[\textit{\underline{Variance Reduction}}] \label{prop:t1}
Generalized antithetic samples produced by Algorithm \ref{alg:1} lead to maximal negative
correlation in (\ref{eq:10}) across sets of $n^\prime$ weight pairs of matching magnitudes.
Sample variance is guaranteed to be less than the sample variance of vanilla IS+SS FRW \cite{Yu2013a}
for approximately the same number of total positive and total negative weight values
on the Gaussian surface.
Approximate equality of total positive vs total negative weight values always holds for IS+SS FRW.
\end{theorem}

This result applies to the actual Monte Carlo estimator from finite samples,
and not just the theoretical expectations of (\ref{eq:9}).
Dielectrics contained in the first transition domain can be entirely arbitrary.
It has also been extensively validated through real-world cases, such as those presented in section \ref{sec:results},
where variance reduction vs IS+SS FRW is indeed quite pronounced.

\section{Numerical Experiments} \label{sec:results}

\subsection{Setup} \label{sec:setup}
Algorithm \ref{alg:1} and the associated variance calculation of (\ref{eq:11}) were implemented within our existing FRW solver and tested
in a series of design cases.
Our primary objective is to measure performance, i.e., extraction speed, and not extraction precision per se.
Relative precision should of course be maintained in all cases.
The test environment used throughout was a Linux machine with Intel Xeon Gold 6326 CPU @ 2.90GHz.
All experiments employed a single thread to better capture experienced speedups.
Note that the method proposed in this paper does not in any way affect the embarassingly parallel nature of the core FRW method.

The particular variants of FRW estimators compared were:

\begin{itemize}
    \item \textbf{IS+SS}: This is the baseline estimator, based on \cite{Visvardis2023}, and it employs importance sampling and stratified sampling on the first transition domain.
    \item \textbf{SMS-n}: This is the multiple-shooting estimator of \cite{Huang2024} employing $n={2,4,8}$ points on each first transition domain selected symmetrically with respect to the VGS.
    \item \textbf{GAS-n}: This is the GAS estimator proposed in this paper, employing $n={2,4,8}$ points on each first transition domain, selected according to Algorithm \ref{alg:1}.
\end{itemize}

In various metrics that follow, we need to estimate the value $p_m = p$ from Lemma \ref{prop:l1} for a \textit{discretized} first transition domain.
The surface of the domain is discretized into $6 * K * K$ square cells ($K * K$ per cube side), with $K$ a user-chosen discretization constant.
Then, with excitation at the center of each of those surface cells, we solve for the gradient of the Green's function
at the center of the domain and compute the inner product of this gradient vector with $\vv{\mathbf{n}}$.
The solver uses standard methods, e.g., FDM, \cite{Yu2013, Visvardis2023}.
Thus we get $6 * K * K$ values which are partitioned into the distinct sets $S^{+}, S^{-}$ according to their sign.
The probability $p$ is then the ratio of the sum of positive inner products over the sum of the absolute values of all inner products.
Note that the above procedure only pertains to estimating the value of $p$ for particular metrics: the GAS schemes need only sample according to $p$, specifically
by Algorithm \ref{alg:1}.

On the other hand, the general premise upon which, for example, the SMS-2 scheme relies on is,
for any sample on the surface of the domain, namely $\mathbf{r_1}^{+}$, to identify its symmetric on the surface of the domain with respect to the Gaussian surface, i.e., with respect to $\vv{\mathbf{n}}$,
namely $\mathbf{r_1}^{-}$, and then form an estimator of the form \cite{Huang2024}:

\begin{equation}
    \begin{split}
        \hat{\mu}_{SMS} = F \cdot \frac{1}{n}(\sum\limits_{i=1}^{n}2 \alpha_{i} w_i^* \phi_{i}^* + 2 (1-a_{i}) \tilde{w}_i^* \tilde{\phi}_{i}^*) \\
        \alpha_{i} = \frac{\lvert\nabla_\mathbf{r} P(\mathbf{r},\mathbf{r_1^{+}}) \cdot \vv{\mathbf{n}}(\mathbf{r})\rvert}{\lvert\nabla_\mathbf{r} P(\mathbf{r},\mathbf{r_1^{+}}) \cdot \vv{\mathbf{n}}(\mathbf{r})\rvert + \lvert\nabla_\mathbf{r} P(\mathbf{r},\mathbf{r_1^{-}}) \cdot \vv{\mathbf{n}}(\mathbf{r})\rvert} \label{eq:13}
    \end{split}
\end{equation}

\noindent where $w_i^*, \phi_{i}^*$ are the weight value and destination potential of the initial sample
and $\tilde{w}_i^*, \tilde{\phi}_{i}^*$ the weight value and destination potential of the symmetric sample.
The above equations should be contrasted against equations (\ref{eq:7}) and (\ref{eq:8}) for GAS, and, apart from the the way the samples, i.e.,
the weight values, are selected on the surface of the first transition domain, there is obvious mathematical similarity.
The most favorable value for $a$ is 0.5 \cite{Huang2024}, however it does vary significantly away from 0.5 in practice, as will be discussed in section \ref{sec:discussion}.
The probability $p$ does not itself appear in the estimator, and is theoretically identical to 0.5, and practically very close to it, as will be discussed in \ref{sec:discussion} and the appendix.
In general, both scheme families (GAS- and SMS-) utilize the inner products of gradient vectors with the Gaussian surface normal vector, however,
whereas the GAS schemes look at the relative sizes of \textit{aggregations} of inner products,
the SMS schemes look at the relative sizes of \textit{specific} inner products.

It is entirely possible functionally to take $N$ antithetic pairs per first transition domain with Algorithm \ref{alg:1}, instead of just one, and these will
be by construction independent to one another. Then, the total sample covariance will be a sum of covariance terms of the form (\ref{eq:11}).
GAS allows any multiple of 2 as number of samples on the surface of each first transition domain (e.g., 4, 6, 8, 10), as opposed to
SMS that allows only powers of 2 (e.g., 4, 8, 16, 32) and thus GAS is more flexible in comparison.

However, an important caveat for both multiple shooting schemes (SMS and GAS) is that the number of antithetic pairs per first transition domain should not be increased indefinitely.
In practice, the value of $F$ from (\ref{eq:8}) still needs to be estimated through Monte Carlo integration, requiring a sufficient number of samples over \textit{the whole} VGS.
Starting a large number of walks per single first transition domain may limit distinct transition domains, i.e., distinct samples on the VGS.
We refer the interested reader to \cite{Huang2024} for a thorough discussion regarding the balance of samples on the VGS and on first transition domains.

For this reason, both SMS and GAS, can benefit from utilizing a joint convergence criterion for both weight values and $F$ \cite{Huang2024},
or even pre-estimating the value of $F$ by itself through a separate Monte Carlo procedure.
As our primary goal is to evaluate performance benefits over the current state-of-the-art, for our main results we have limited the number of samples per first transition domain to 8,
without utilizing such a special convergence criterion.

\subsection{Evaluation} \label{sec:evaluation}
We tested against 9 design cases, roughly divided into two categories: designs without LDE and designs with LDE.
The non-LDE design cases (cases 1-4) included simpler structures such as a capacitor and inductor (cases 1-2),
as well as more complicated IC designs, e.g., with interposers (cases 3 and 4), all under contemporary technology nodes.
The LDE design cases (cases 5-9) were structures of the same vein and technology node as utilized in \cite{Visvardis2023} (5nm)
, i.e., metal wires placed above a conductive plane, or between conductive planes,
all affected specifically by loading, thickness variation and dielectric damage.
Note that these are exactly the effects that would be applied in real-world VLSI processes
and create for example dielectric inhomogeneities.
All relevant effects are applied identically before extraction for all FRW variants.
Of these, cases 5-7 did not have a top conductive plane, whereas cases 8-9 did.
Aggregate statistics for the test cases are shown in Table \ref{tbl:1}.

\begin{table*}
    \begin{center}
    \caption{Test cases statistics}
    \label{tbl:1}
        \begin{tabular}{| c | c | c | c |}
            \toprule
            Case & \#nets & \#blocks & $\#blocks_{m}$ \\ \hline
            1 & 2 & 18 & 18 \\
            2 & 1 & 3595 & 3595 \\
            3 & 102 & 51491 & 51491 \\
            4 & 20 & 22711 & 22711 \\
            5 & 5 & 116 & 4 \\
            6 & 5 & 742 & 4 \\
            7 & 5 & 468 & 4 \\
            8 & 5 & 571 & 8 \\
            9 & 5 & 1209 & 5 \\
            \bottomrule
        \end{tabular}
    \end{center}
\end{table*}

The termination criterion for net extraction in the experiments was always the relative standard error, i.e., the ratio of sample standard deviation over sample mean, to be less than 0.5\% (i.e., 0.005).
Convergence was checked every 1000 completed walks (referred to as batch size), to achieve better granularity over possible differences in walks.
Several independent extractions for each design case were ran to mitigate the inherent stochasticity of the FRW method.
The LDE cases utilized the machine learning models for Green's function data from \cite{Visvardis2023} to facilitate extraction.

Besides wall times, another performance metric examined was the number of first transition domains formed until convergence,
which we consider to be a more objective metric of variance reduction for these types of approaches, all other things being equal,
because it directly measures how many samples are necessary to achieve the same level of accuracy by a Monte Carlo estimator.
Apart from the way samples were selected on the surface of the first transition domains,
all other infrastructure for our implementation of the variants was identical, including Green's function data through precomputation and machine learning models.
Hence we do not expect hops to differ substantially among the variants and focus on total number of first transition domains.
The latter has a straightforward interpretation in terms of walks as well: for IS+SS it is exactly the number of walks, whereas for
SMS-n and GAS-n walks it can be found by multiplying the number of first transition domains by $n$.

\begin{table*}
    \begin{center}
        \caption{Results for Baseline}
        \label{tbl:2}
        \begin{tabular}{| c | c | c | c |}
            \toprule
                & \multicolumn{3}{ c |}{IS+SS} \\
            \cline{2 - 4}
            Case & $C_{tot}(aF/\mu m)$ & \#first domains & time(sec) \\
            \midrule
                1 & 11.79 & 604600 & 18.1 \\
                2 & 195.9 & 1798000 & 89.4 \\
                3 & 1700317.74 & 42750100 & 5046 \\
                4 & 10099.06 & 83884000 & 10062 \\
                5 & 25.93 & 301300 & 334.2 \\
                6 & 17.38 & 159700 & 413.3 \\
                7 & 33.54 & 489300 & 324.2 \\
                8 & 15.37 & 167000 & 295.1 \\
                9 & 24.26 & 77100 & 114.8 \\
            \bottomrule
        \end{tabular}
    \end{center}
\end{table*}

\begin{table*}
    \begin{center}
        \caption{Results for SMS-2 and GAS-2}
        \label{tbl:3}
        \begin{tabular}{| c | c | c | c | c | c | c |}
            \toprule
                & \multicolumn{3}{ c |}{SMS-2}
                & \multicolumn{3}{ c |}{GAS-2} \\
            \cline{2 - 7}
            Case & $C_{tot}(aF/\mu m)$ & \#first domains & time(sec) & $C_{tot}(aF/\mu m)$ & \#first domains & time(sec) \\
            \midrule
                1 & 11.45 & 324000 & 15.3 & 11.78 & 298700 & 14.5 \\
                2 & 191.2 & 899000 & 63.7 & 196.44 & 899000 & 61.2 \\
                3 & 1650610.05 & 11820200 & 2504.9 & 1698504.21 & 13181000 & 2643.2 \\
                4 & 9753.36 & 12622700 & 3572.7 & 10076.25 & 12091100 & 3061.9 \\
                5 & 26.33 & 163400 & 250.3 & 25.91  & 148200 & 220.9 \\
                6 & 17.27 & 95100  & 378.9 & 17.35 & 79200 & 317.1 \\
                7 & 32.97 & 259300 & 234.4 & 33.46 & 245400 & 208.5 \\
                8 & 15.16 & 106700 & 264.3 & 15.38  & 82700 & 204.1 \\
                9 & 23.91 & 53100 & 100.1 & 24.27 & 38300 & 75.3 \\
            \bottomrule
        \end{tabular}
    \end{center}
\end{table*}

\begin{table*}
    \begin{center}
        \caption{Results for SMS-4,8 and GAS-4,8}
        \label{tbl:4}
        \scalebox{0.8}{
        \begin{tabular}{| c | c | c | c | c | c | c | c | c |}
            \toprule
                & \multicolumn{2}{ c |}{SMS-4}
                & \multicolumn{2}{ c |}{GAS-4}
                & \multicolumn{2}{ c |}{SMS-8}
                & \multicolumn{2}{ c |}{GAS-8} \\
            \cline{2 - 9}
            Case & \#first domains & time(sec) & \#first domains & time(sec) & \#first domains & time(sec) & \#first domains & time(sec) \\
            \midrule
                1 & 183300 & 15.4
                  & 150500 & 13.2
                  & 103600 & 16.3
                  & \textbf{75700} & \textbf{12.4} \\

                2 & 450000 & 50.9
                  & 450000 & 49.4
                  & \textbf{225000} & 43.1
                  & \textbf{225000} & \textbf{42.3} \\

                3 & 7150700 & 2662.4
                  & 6895800 & 2403.1
                  & 4165400 & 2795.1
                  & \textbf{3650000} & \textbf{2304.2} \\

                4 & 7196600 & 4107
                  & 6539800 & 2904.2
                  & 3801700 & 4000.9
                  & \textbf{3292600} & \textbf{2775.7} \\

                5 & 92000 & 218.1
                  & 74500 & 166.7
                  & 49400 & 195.4
                  & \textbf{37700} & \textbf{139.6} \\

                6 & 56200 & 380.9
                  & 40400 & 269.8
                  & 33700 & 423.6
                  & \textbf{20000} & \textbf{238.8} \\

                7 & 144800 & 196.1
                  & 122700 & 160.3
                  & 76200 & 171.3
                  & \textbf{61300} & \textbf{129.6} \\

                8 & 61900 & 237.9
                  & 42000 & 158.3
                  & 36600 & 248
                  & \textbf{21400} & \textbf{140.3} \\

                9 & 31900 & 92.3
                  & 19100 & 58.9
                  & 20000 & 98.5
                  & \textbf{10000} & \textbf{53.4} \\
            \bottomrule
        \end{tabular}
        }
    \end{center}
\end{table*}

The main results of experiments for all FRW variants are shown in Tables \ref{tbl:2}, \ref{tbl:3} and \ref{tbl:4}.
The number of first transition domains reported is the sum of first transition domains to extract all nets in the design.
Total capacitance $C_{tot}$ reported is the sum of cross capacitances towards conductors of different nets for all nets in the design.
Due to space considerations, $C_{tot}$ is only reported for SMS-2 and GAS-2, since extracted values do not substantially differ for the other variants.
Accuracy is acceptable with the multiple shooting schemes compared to IS+SS for the number of antithetic pairs employed.
Additionally, note that the flavor of IS+SS used as baseline has been shown in \cite{Visvardis2023} to closely follow the accuracy of a reference tool (Raphael)
for designs with LDE.

We have experienced substantial performance gains with SMS-n, GAS-n over IS+SS in virtually every design we have tried.
They are quite faster than IS+SS for the large cases 3 and 4, while also being noticeably faster in the simpler cases 1 and 2.
In cases 3 and 4 in particular, SMS-2 achieves speedups of 2x and 2.8x respectively over IS+SS.
Our proposed GAS-2 scheme follows suit with only a slight lag behind SMS-2 in case 3 and a slight edge over SMS-2 in case 4.
For cases 1-4, GAS-2 and SMS-2 are practically very close performance-wise.

On the other hand, there is a distinct advantage of our proposed GAS-n scheme in cases 4-9 over SMS-n,
in terms of first transition domains necessary for convergence,
ranging from \textbf{19\%} for case 7 to \textbf{50\%} for case 9, both with GAS-8 vs SMS-8.
In terms of wall times, the advantage ranges from 28.5\% to 45.8\%,
translating to \textbf{~1.32x} and \textbf{~1.84x} speedups respectively.

The best wall times for all cases are achieved by the GAS-8 variant.
For the SMS schemes, extraction time deteriorates when moving from SMS-2 to SMS-4 in cases 1, 3, 4 and 6 and when moving from SMS-4 to SMS-8 in cases 8 and 9.
Figure \ref{fig:speedup} summarizes the \textit{adversarial} speedup of GAS-8 against the best available SMS variant for all cases.
This is effectively a worst-case comparison for GAS-8, since in practice selecting the most suitable SMS variant would require some kind of adaptive scheme,
e.g., similar to what is described briefly in \cite{Huang2024}, and would itself add to the experienced runtime.

For reasons alluded to in section \ref{sec:setup}, performance cannot be expected to improve indefinitely by increasing the number of antithetic pairs per first transition domain.
Using more than four antithetic pairs per first transition domain (GAS-8) is possible in a case-by-case basis, but only judiciously,
for example by tightening the standard convergence criterion, utilizing a joint weight values and $F$ convergence criterion \cite{Huang2024},
or even pre-estimating $F$ by itself before any random walks are launched.

\begin{figure*}[h!]
    \includegraphics[width=.75\textwidth]{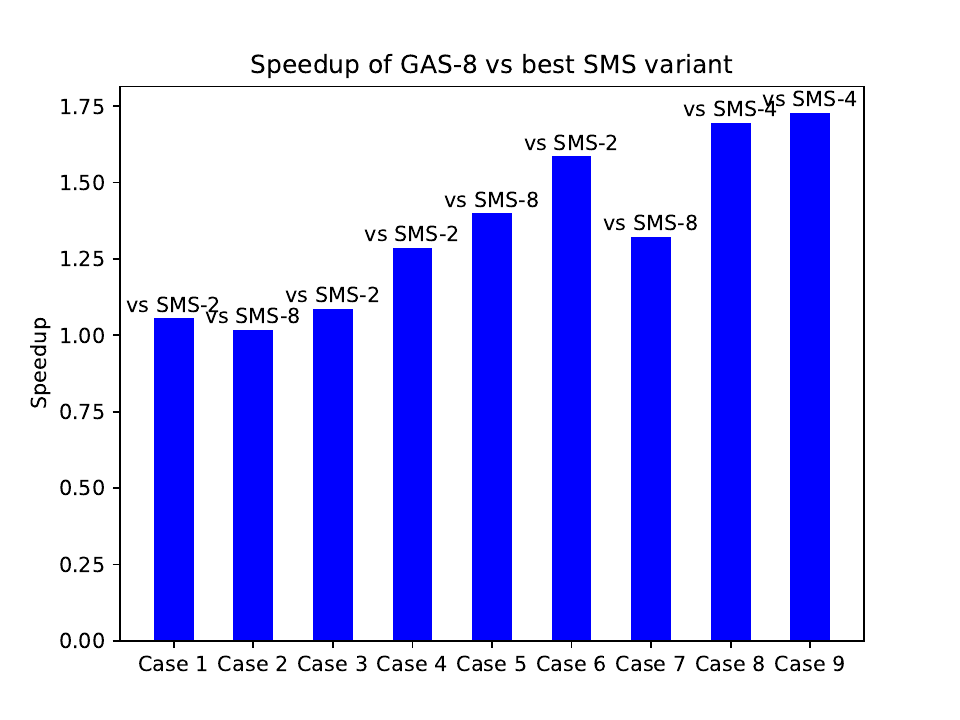}\hfill
    \caption{Time speedup of GAS-8 vs the best SMS scheme for each test case}
    \label{fig:speedup}
\end{figure*}

\subsection{Discussion} \label{sec:discussion}

Fig. \ref{fig:p_stats} reports the full distribution of $p$ across all transition domains for each test case, while Table \ref{tbl:5} reports its average and standard deviation.
It is evident that actual $p$ values encountered are very close to the ideally desirable 0.5, even with discretization, in fact, not a single first transition domain out of millions in total across all extractions, possesses a $p$ value outside the range 0.45-0.55.
In the non-LDE cases, where domains do not in general contain complex dielectric configurations and the FDM solution to $\nabla_\mathbf{r} P(\mathbf{r},\mathbf{r_1})$ is be expected to be more accurate, the standard deviation of $p$ away from 0.5 is overwhelmingly close to 0.
These findings, together with variance reduction from (\ref{eq:11}), can account for the superior performance of the GAS method.

\begin{figure*}[h!]
    \includegraphics[width=.25\textwidth]{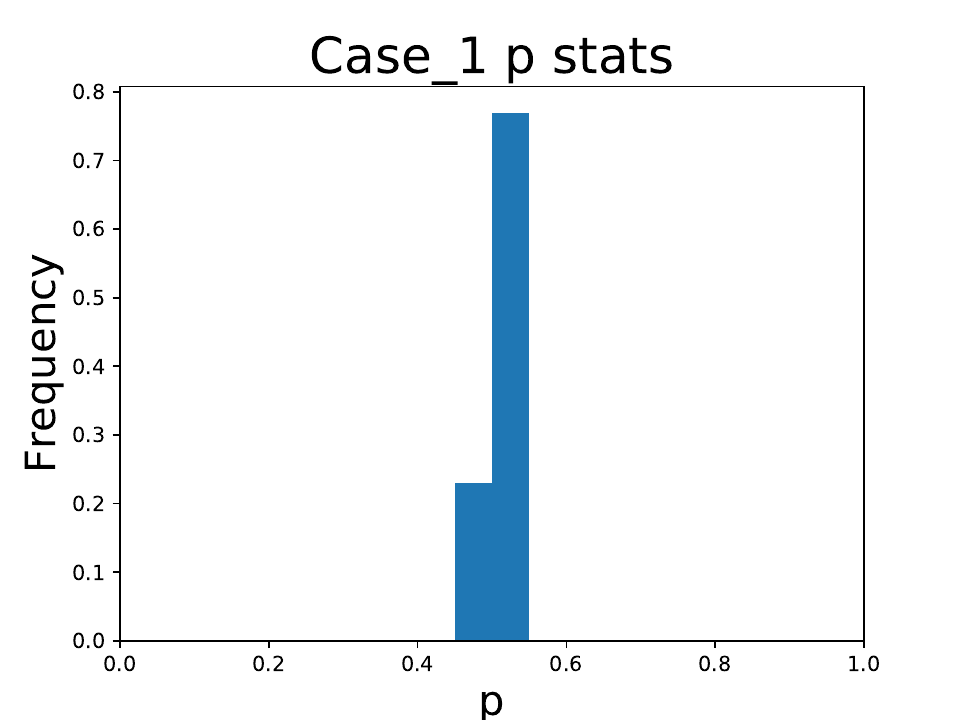}\hfill
    \includegraphics[width=.25\textwidth]{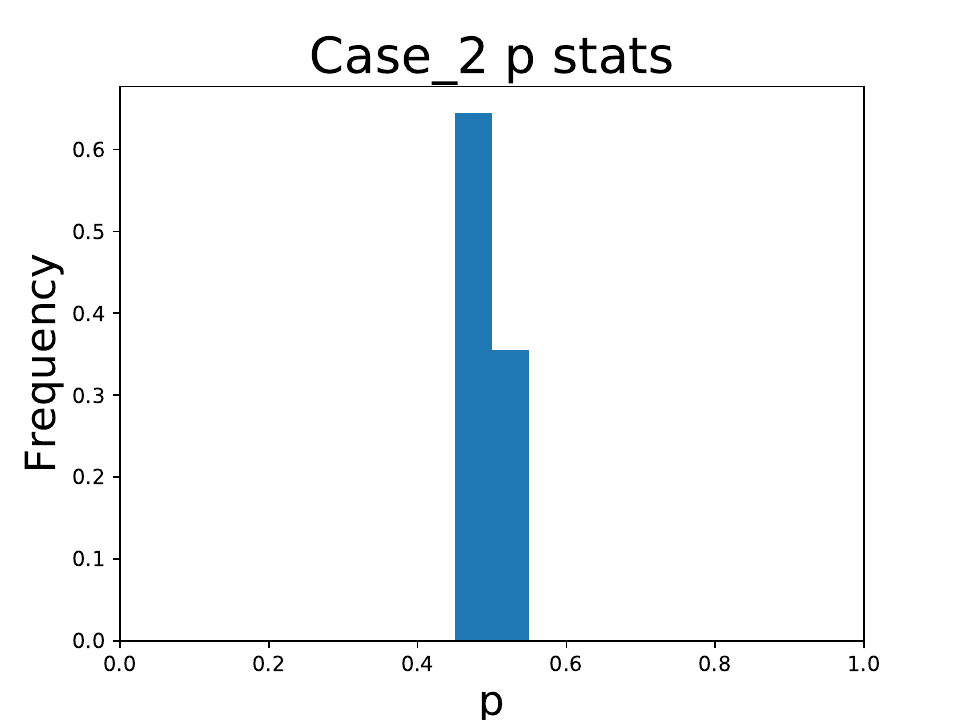}\hfill
    \includegraphics[width=.25\textwidth]{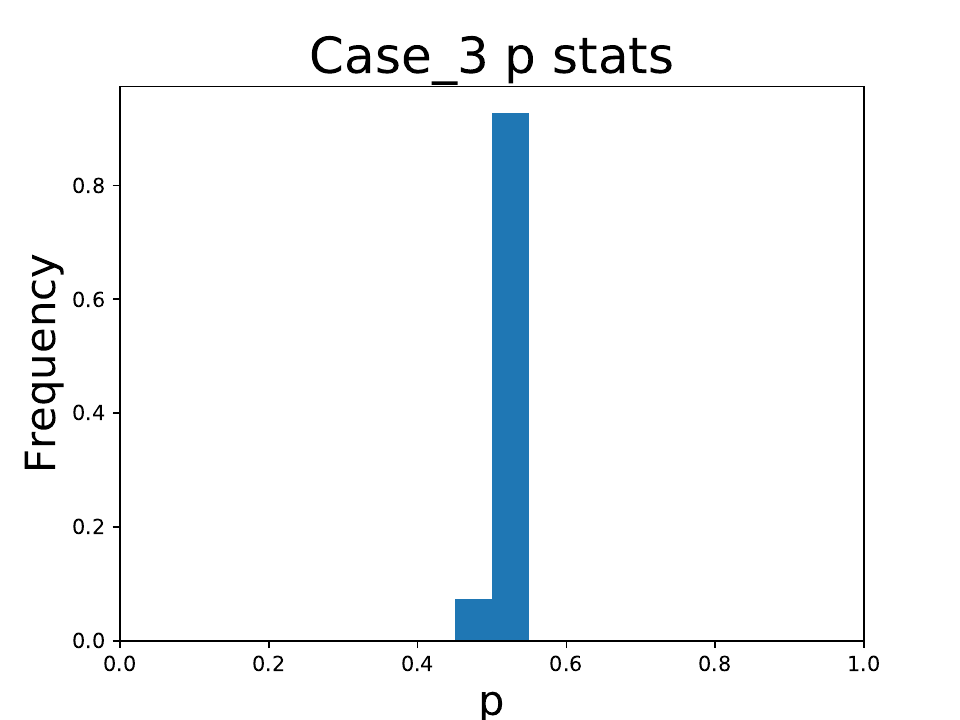}\hfill
    \includegraphics[width=.25\textwidth]{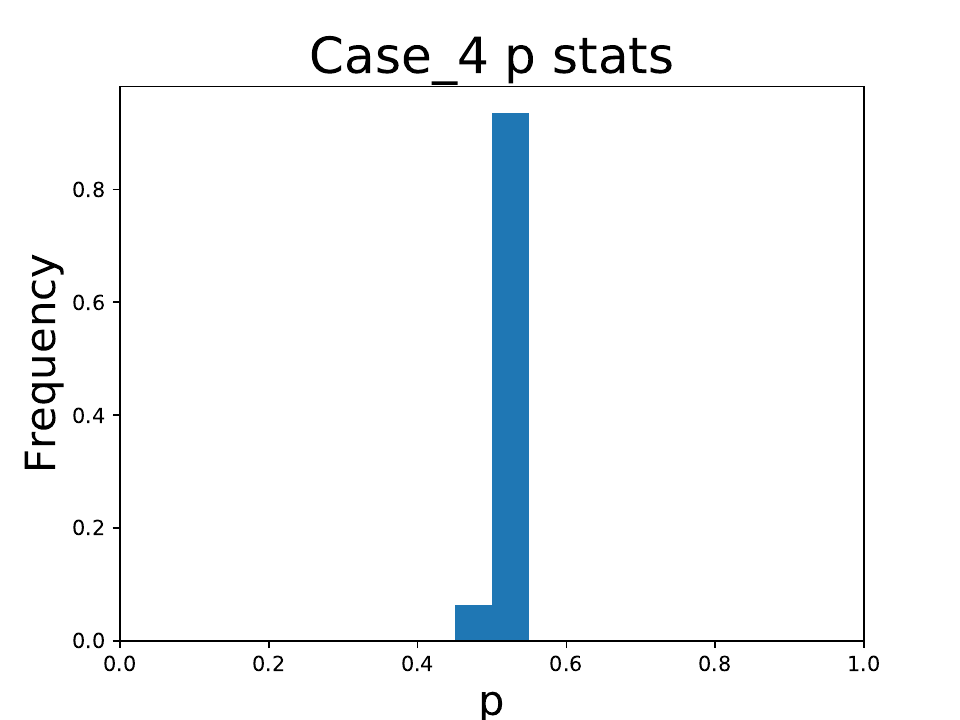}\hfill
    \includegraphics[width=.2\textwidth]{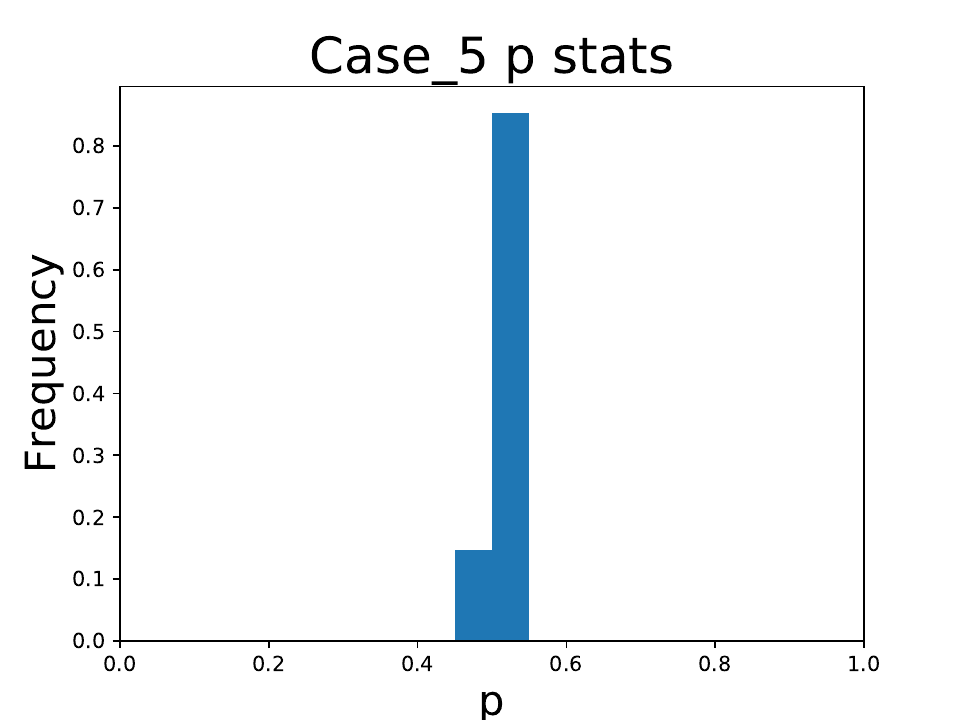}\hfill
    \includegraphics[width=.2\textwidth]{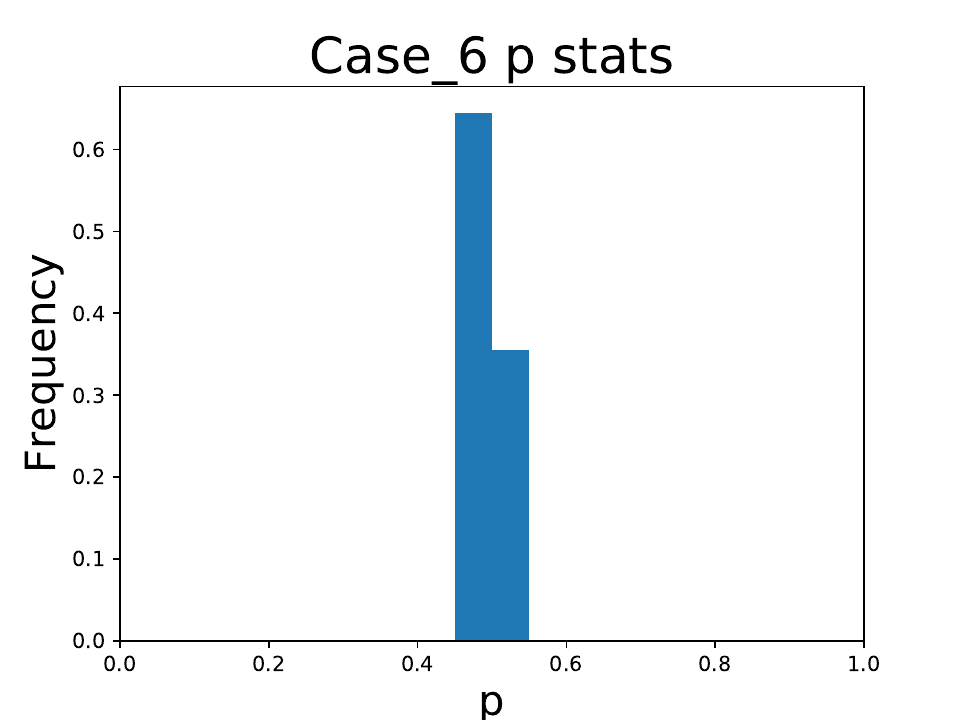}\hfill
    \includegraphics[width=.2\textwidth]{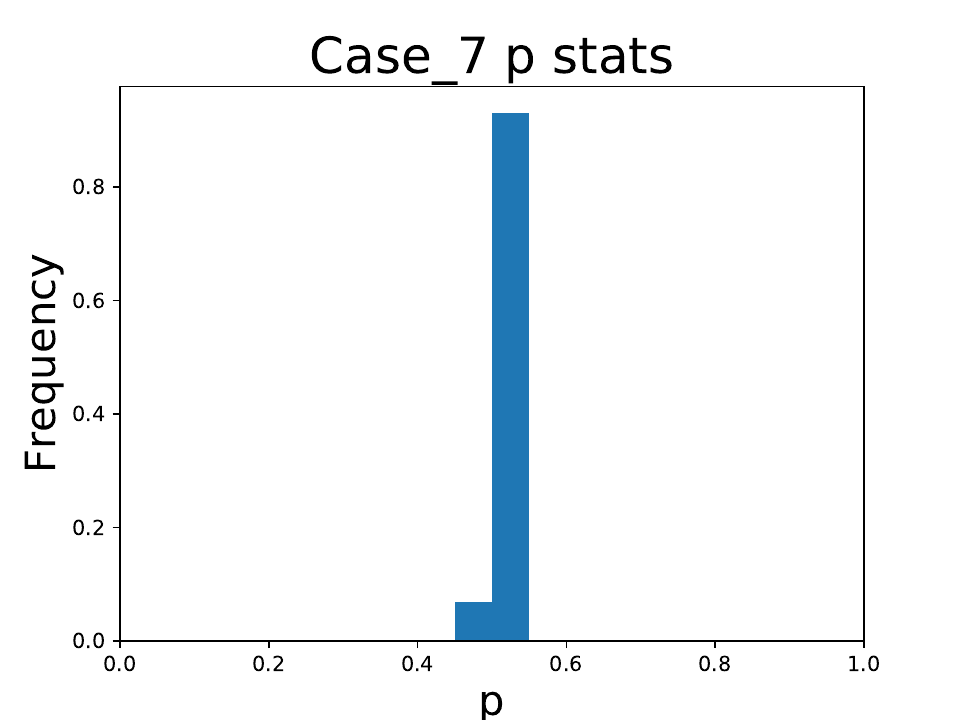}\hfill
    \includegraphics[width=.2\textwidth]{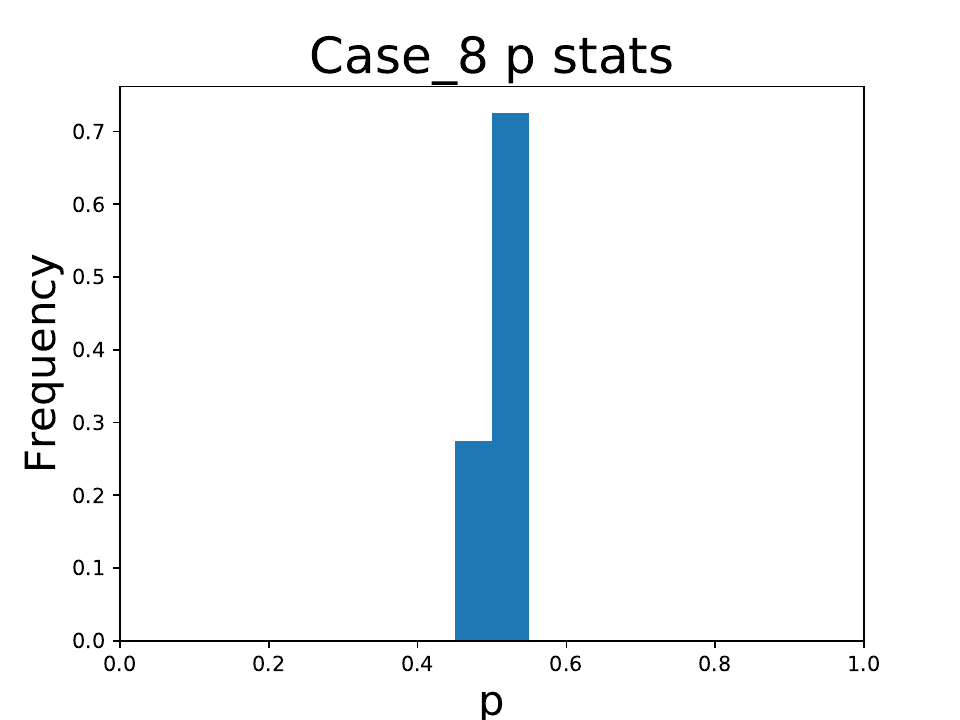}\hfill
    \includegraphics[width=.2\textwidth]{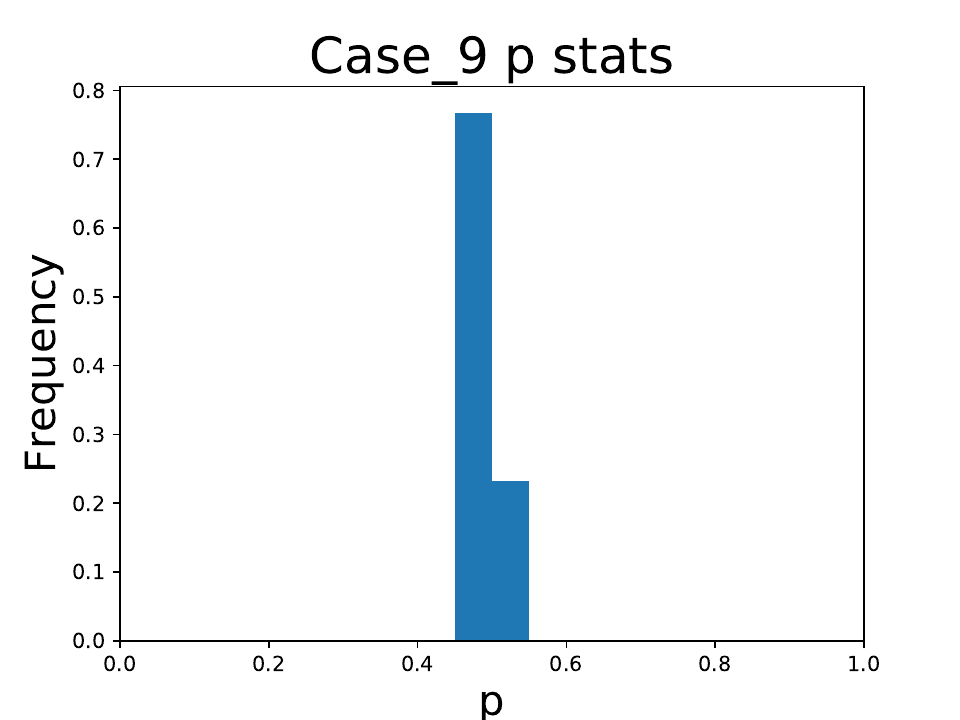}\hfill
    \caption{Histogram of the probability p of equation (\ref{eq:12}) for all test cases}
    \label{fig:p_stats}
\end{figure*}

On the other hand, SMS variance in \cite{Huang2024} uses a different mathematical formulation, where negative correlation among the samples is an emerging, not foundational property.
Examining SMS-2 and GAS-2 in particular, to better isolate the essential idea of each method, one shortcoming of SMS-2 is it has no way of guaranteeing that taking symmetric points with respect to
the Gaussian surface will lead to a negative sample covariance in (\ref{eq:11}).
This could happen primarily in cases such as 5-9, where the Green's function itself exhibits asymmetries, because ambient dielectrics are not necessarily stratified and thus do not profess xy-plane rotational invariance.
It could also happen in designs with stratified dielectrics, but now on the top and bottom sides of the Gaussian surface,
because there is no symmetry along the z-axis.
A second related issue of SMS-2 is that each weight utilizes an extra multiplicative factor $a$
necessary to make the final estimator unbiased, as per (\ref{eq:13}).
These factors are not fixed across first transition domains and can increase final variance.
GAS-2 by construction utilizes a fixed weight factor of 0.5 to achieve unbiasedness, as per (\ref{eq:7}).

\begin{table*}
    \begin{center}
    \caption{Statistics of GAS-2 vs SMS-2}
    \label{tbl:5}
        \begin{tabular}{| c | c | c | c | c | c |}
            \toprule
                Case
                & \multicolumn{1}{p{2cm} |}{\centering \% same sign weight pairs (SMS-2)}
                & \multicolumn{1}{p{2cm} |}{\centering mean a (SMS-2)}
                & \multicolumn{1}{p{2cm} |}{\centering std deviation of a (SMS-2)}
                & \multicolumn{1}{p{2cm} |}{\centering mean p (GAS-2)}
                & \multicolumn{1}{p{2cm} |}{\centering std deviation of p (GAS-2)} \\
            \midrule
            1 & 0.08317757 & 0.499967 & 0.0301019 & 0.5      & 1.29e-10 \\
            2 & 0.16607342 & 0.499813 & 0.0541569 & 0.5      & 9.47e-10 \\
            3 & 0.31535203 & 0.500172 & 0.0454918 & 0.5      & 3.94e-10 \\
            4 & 0.20896187 & 0.500246 & 0.0662073 & 0.5      & 1.57E-10 \\
            5 & 3.53641975 & 0.494368 & 0.126159  & 0.500123 & 0.00146554 \\
            6 & 1.62553191 & 0.498273 & 0.0982273 & 0.500143 & 0.00263955 \\
            7 & 0.20697674 & 0.502388 & 0.0626808 & 0.500132 & 0.00128501 \\
            8 & 1.28317757 & 0.499478 & 0.109218  & 0.50003  & 0.00332339 \\
            9 & 1.05094340 & 0.49776  & 0.104481  & 0.499999 & 0.00241329 \\
            \bottomrule
        \end{tabular}
    \end{center}
\end{table*}

In Table \ref{tbl:5}, means and standard deviations are reported for the factor $a$ as well.
A small percentage of symmetric pairs of points on the same first transition domains do in fact have the same sign and produce positive products for SMS-2
as shown in Table \ref{tbl:5}.
It can additionally be seen that in cases 5-9 the standard deviation of $a$ appears increased and these are exactly the cases where
LDE are applied.
In cases 1-4 where SMS-2 is competitive, either the percentage of positive correlation products is low, or the standard deviation of $a$ is low, or both.
A full distribution of the weight factors $a$ across all transition domains for each test case is shown in Fig. \ref{fig:a_stats}
and confirms the more pronounced spread in cases 5-9.

\begin{figure*}[h!]
    \includegraphics[width=.25\textwidth]{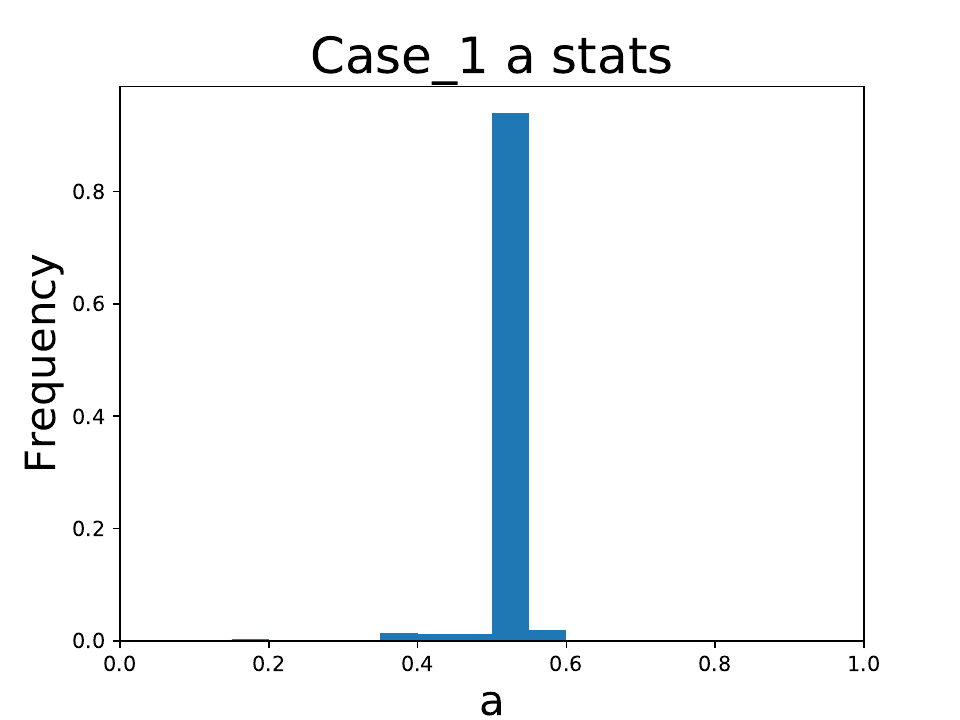}\hfill
    \includegraphics[width=.25\textwidth]{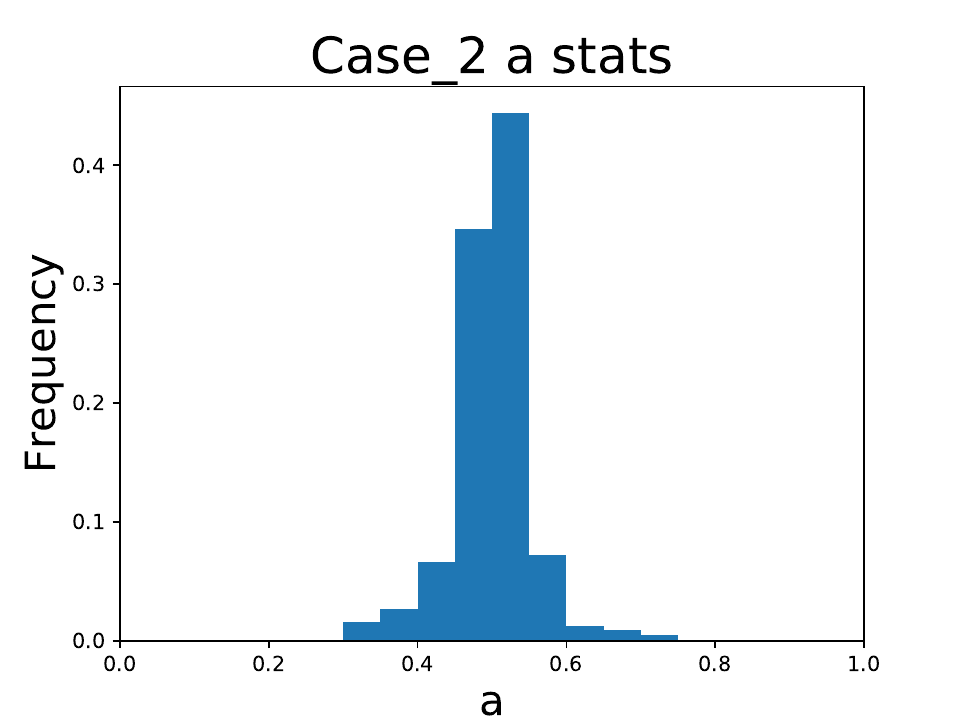}\hfill
    \includegraphics[width=.25\textwidth]{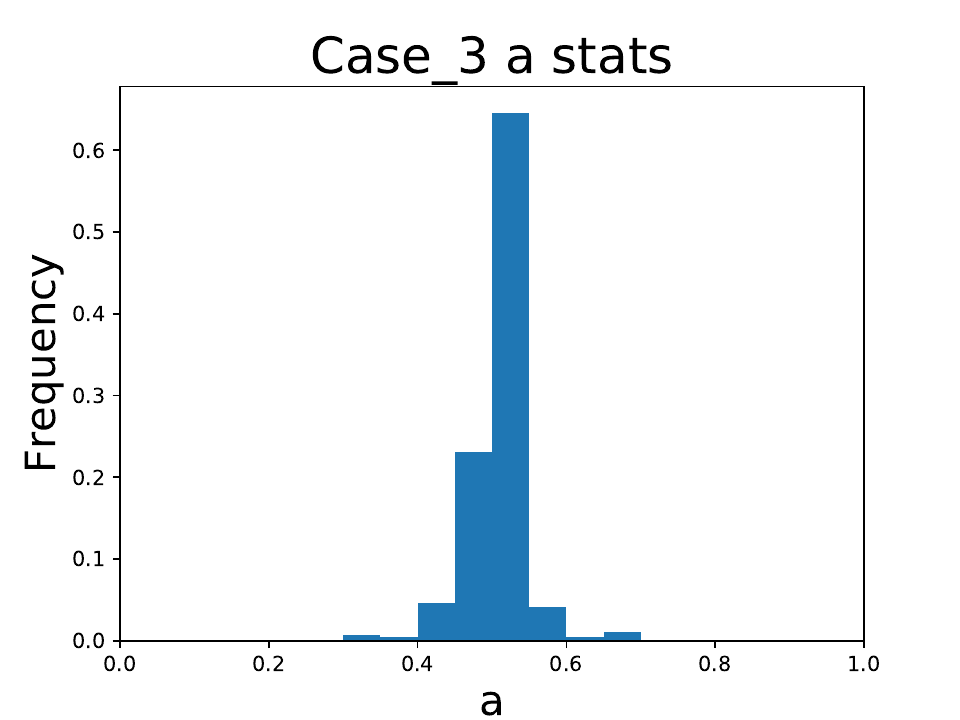}\hfill
    \includegraphics[width=.25\textwidth]{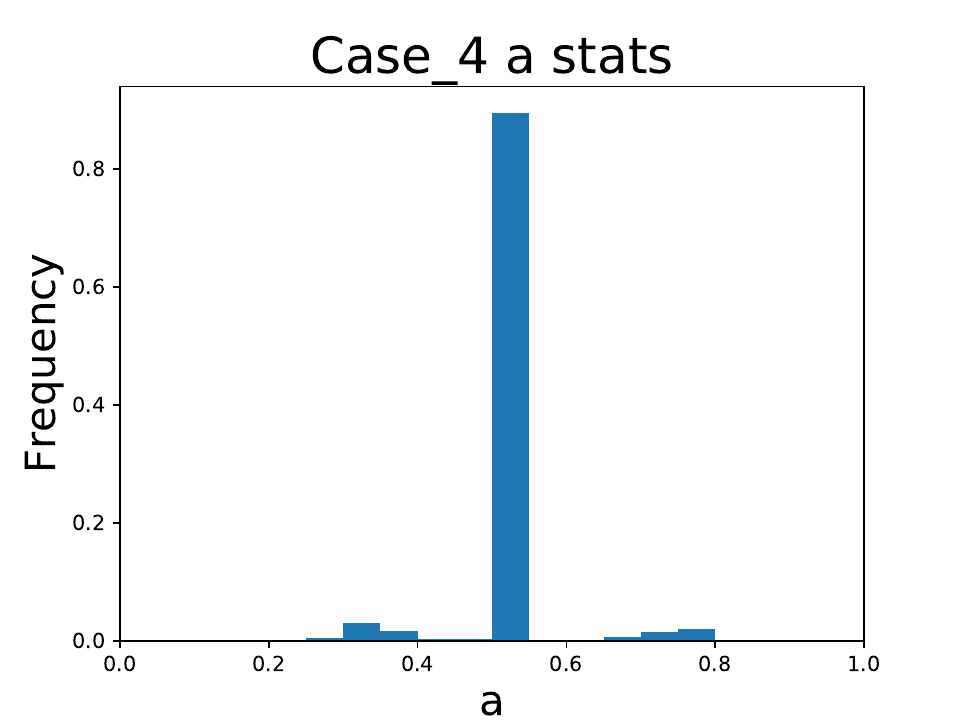}\hfill
    \includegraphics[width=.2\textwidth]{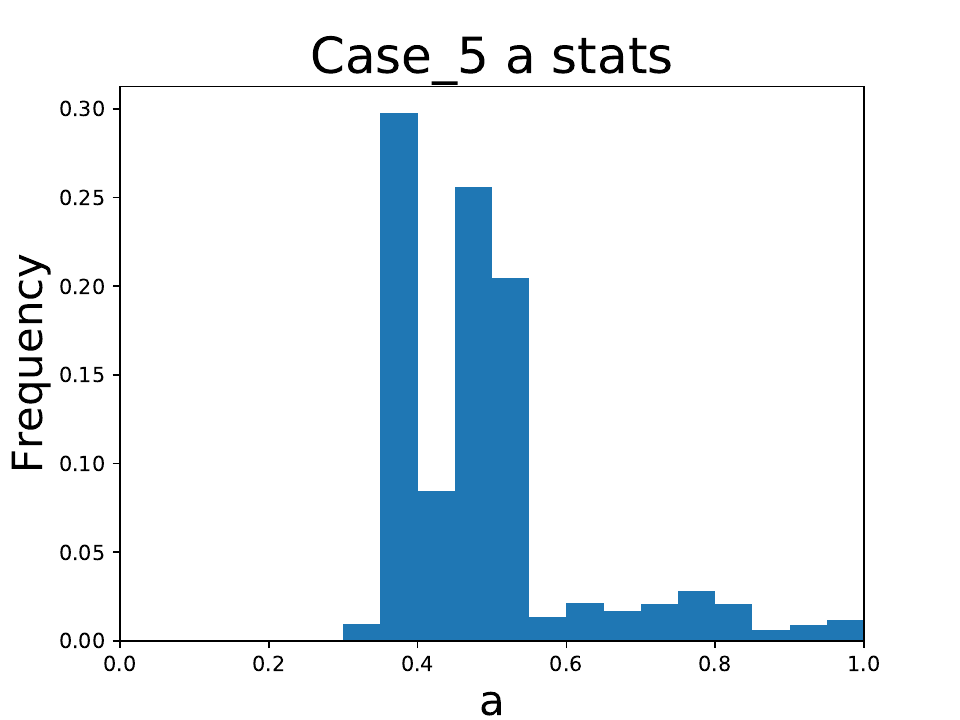}\hfill
    \includegraphics[width=.2\textwidth]{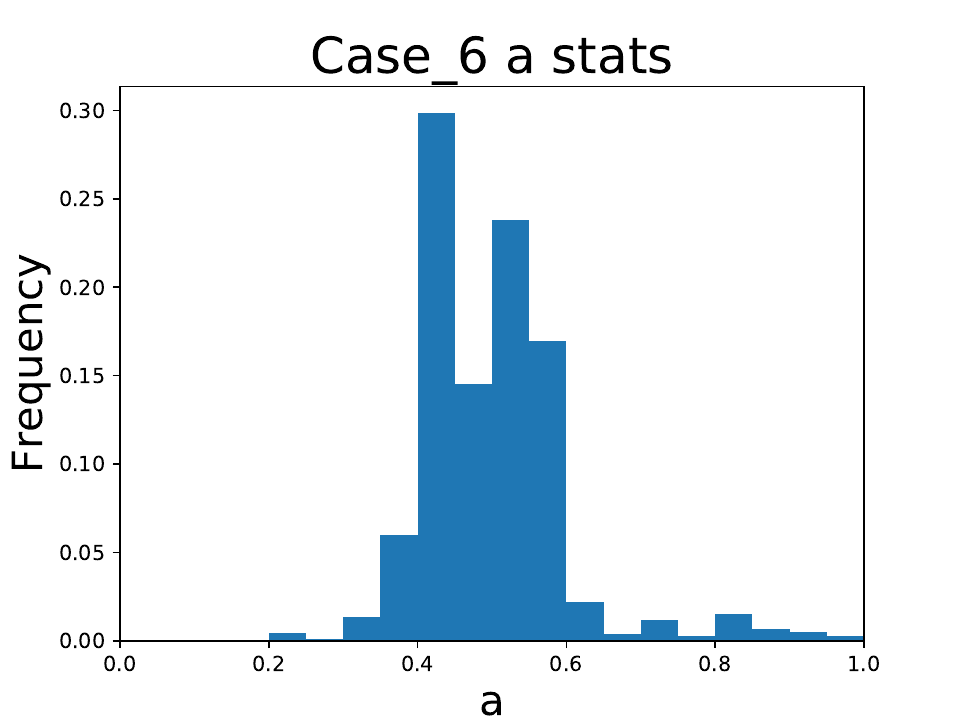}\hfill
    \includegraphics[width=.2\textwidth]{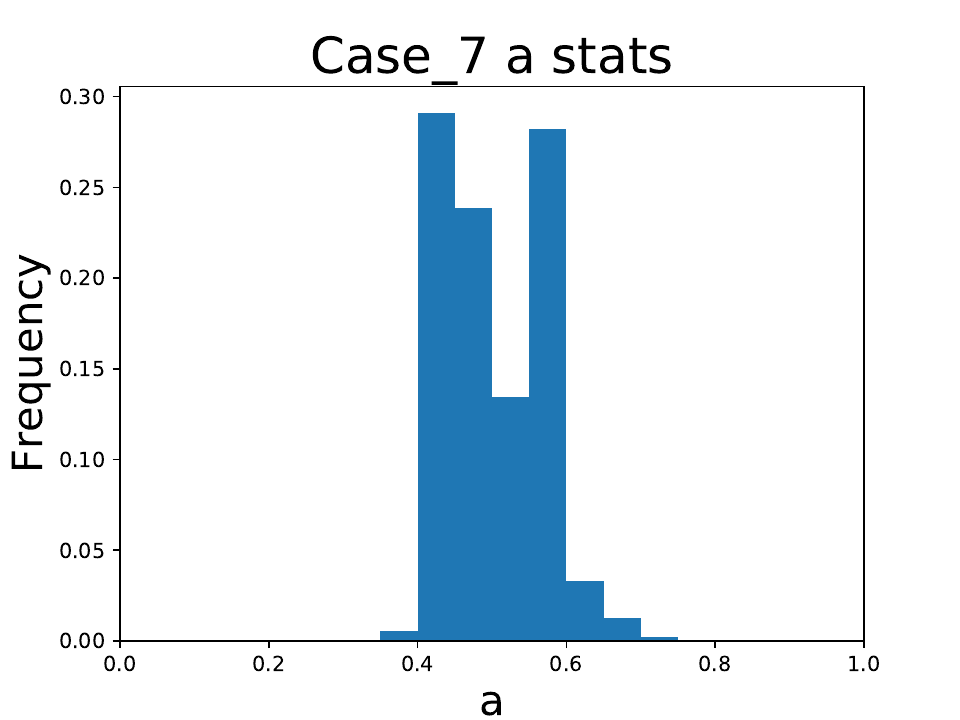}\hfill
    \includegraphics[width=.2\textwidth]{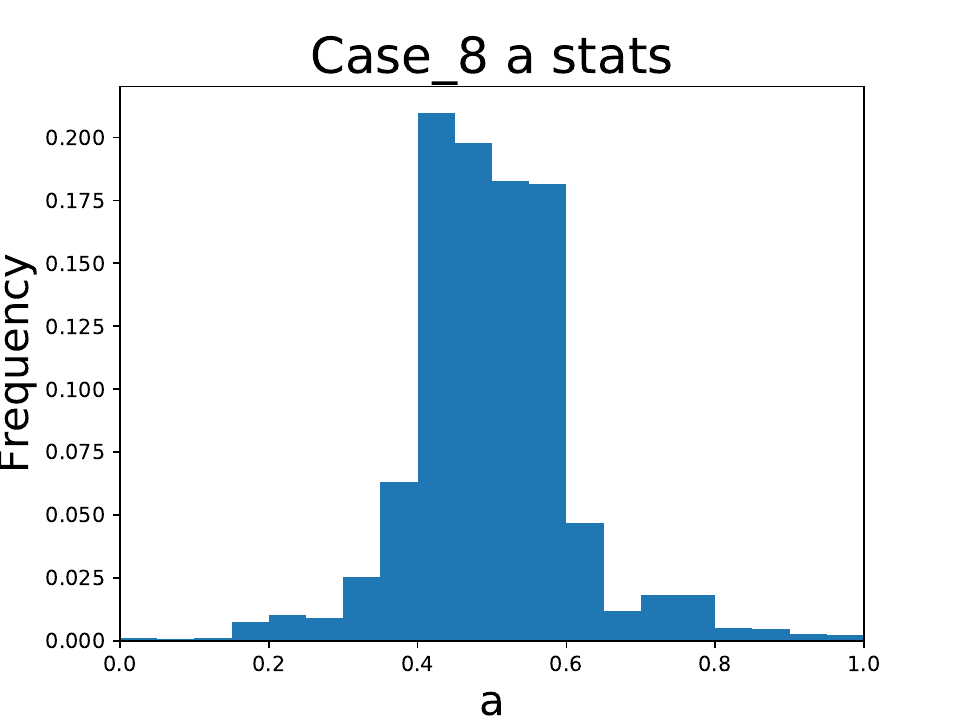}\hfill
    \includegraphics[width=.2\textwidth]{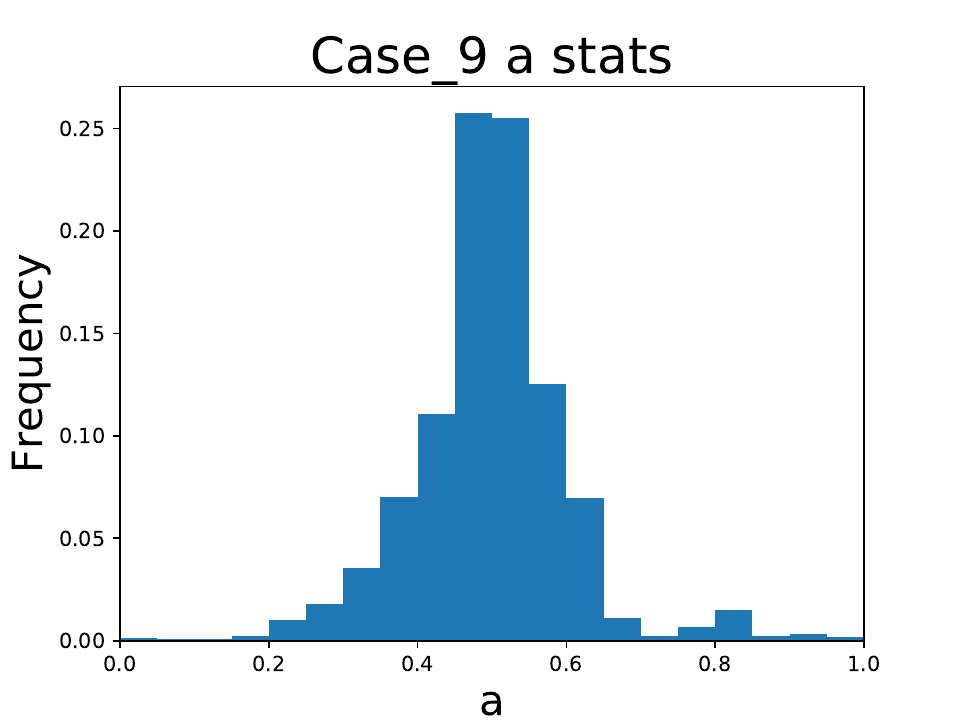}\hfill
    \caption{Histogram of the factor a of equation (\ref{eq:13}) for all test cases}
    \label{fig:a_stats}
\end{figure*}

\begin{figure*}[h!]
    \subcaptionbox{Same sign weight values}{\includegraphics[width=.5\textwidth]{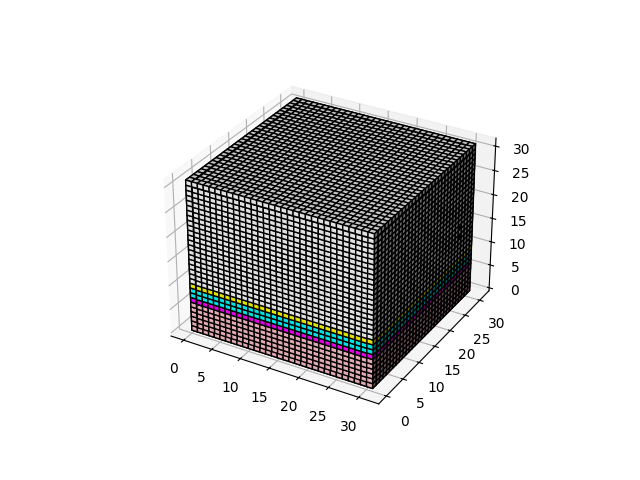}\hfill \includegraphics[width=.4\textwidth]{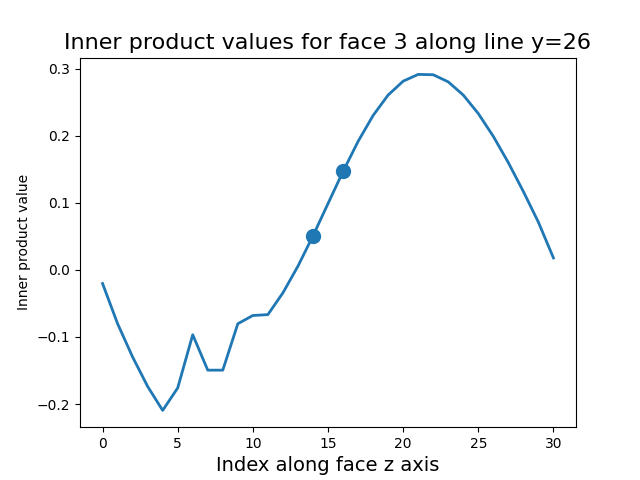}\hfill}
    \subcaptionbox{Factor "a" > 0.95}{\includegraphics[width=.5\textwidth]{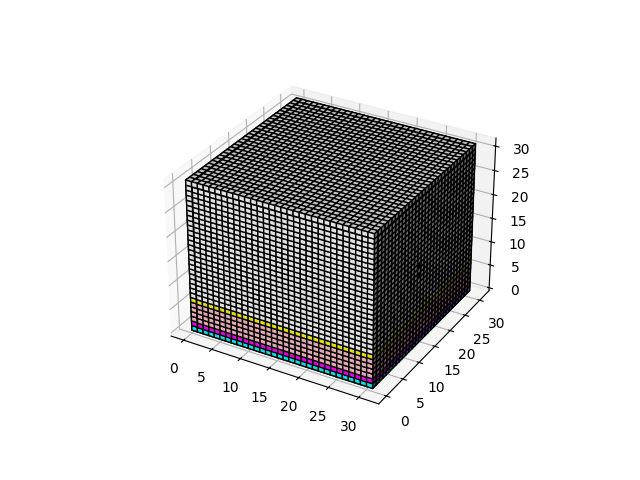}\hfill \includegraphics[width=.4\textwidth]{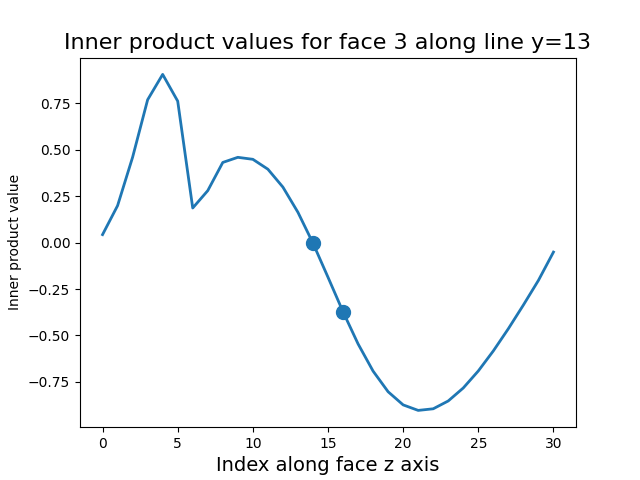}\hfill}
    \caption{First transition domains containing layered dielectrics with problematic samples selected by SMS-2 marked as black.
             Values of $\nabla_\mathbf{r} P(\mathbf{r},\mathbf{r_1})$ along a line parallel to the z-axis shown on the right.}
    \label{fig:concrete_sms_examples}
\end{figure*}

Concrete examples of first transition domains where the above factors come into play are shown in Figure \ref{fig:concrete_sms_examples}, taken directly from Case 8 of the experiments.
Both these examples profess only stratified dielectrics within the transition domain,
with 5 distinct permittivity layers: the values are ranging from 2.3 to 6 in the top, "same sign", case and from 2.3 to 9 in the bottom, "deviant factor" case.
The Gaussian surface is a plane parallel to the x-y plane and passing from the center of the domain while the normal is aligned towards the positive and negative z-axis respectively.
The values of $\nabla_\mathbf{r} P(\mathbf{r},\mathbf{r_1})$ that dictate the weight value signs and $a$ factor values are shown on the right.
It can be seen that the zero crossing of the inner product value along the z axis does not coincide with the plane of spatial symmetry and this creates obvious problems for SMS:
in the top case, both weight values have the same sign and in the bottom case, one inner product is very close to zero while the other is quite larger in magnitude.

To summarize, the SMS schemes, because they rely only on spatial positions and not actual weight values,
could be affected by factors to which the GAS schemes are by construction agnostic, without however compromising efficiency.

\section{Conclusion}
We have presented a novel variance reduction approach for Floating Random Walk-based capacitance extraction.
The approach builds upon the important idea of shooting multiple walks from each first transition domain, which was recently proposed in the literature and improves upon it
by making the selection of points to continue walks from largely data-driven instead of solely geometric.
Selection is guided by the sign of the inner product of the gradient of the surface Green's function and the normal to the Gaussian surface
and pairs of points with opposite such signs are obtained by repeatedly sampling from the same pdf through a Markov chain algorithm.

The approach has been shown to reap substantial performance benefits on top of what the currently best variance reduction approaches can afford,
especially in scenarios where non-stratified dielectrics make pure geometric point selection less suitable.
A thorough investigation of the tradeoff between taking more antithetic pairs on the same first transition domain and
forming more first transition domains in the first place remains as future work,
in order to possibly decide beforehand what the best number of antithetic pairs per transition domain is for any given design.

\section{Acknowledgments}
We would like to thank the anonymous referees both of the current and of previous versions of this manuscript.
Initial comments greatly improved the rigor of our mathematical exposition, while subsequent scrutiny uncovered further mistakes.
Qualitative comments were insightful and enabled us to delve deeper into the subject matter.

\bibliographystyle{ACM-Reference-Format}
\bibliography{TODAES}

\appendix

\section{Markov Chain Analysis}

\begin{figure}[h!]
    \centering
    \usetikzlibrary{automata}
\usetikzlibrary{positioning}
\usetikzlibrary{arrows}

\begin{tikzpicture} [>=stealth', auto, prob/.style={inner sep=1pt, font=\footnotesize}]
    \node (s) [state] {S};
    \node (mp) [state, above right=1em and 4em of s] {M+};
    \node (mm) [state, below right=1em and 4em of s] {M-};
    \node (e) [state, below right=1em and 4em of mp] {T};

    \path[->] (s) edge[bend left=20] node[prob, above=0.3em]{$p$} (mp);
    \path[->] (s) edge[bend right=20] node[prob, below=0.3em]{$1 - p$} (mm);
    \path[->] (mm) edge[bend right=20] node[prob, below=0.3em]{$p$} (e);
    \path[->] (mp) edge[bend left=20] node[prob, above=0.3em]{$1 - p$} (e);

    \path[->] (mp) edge[loop above] node[prob]{$p$} (mp); \path[->] (mm) edge[loop below] node[prob]{$1 - p$} (mm);

\end{tikzpicture}
    \caption{Markov chain for selecting antithetic points (Algorithm \ref{alg:1}). Here $p$ is the probability that the weight value corresponding to a point on the surface of a given first transition domain is positive,
        $\mathbf{S}$ is the start state, $\mathbf{M_{+},M_{-}}$ are the states after first sampling a point that corresponds to a positive or negative weight value respectively,
        and $\mathbf{T}$ is the terminal state of the algorithm.}
    \label{fig:markov-chain}
\end{figure}

The Markov chain for Algorithm \ref{alg:1} is shown in Fig. \ref{fig:markov-chain},
where $p$ is the probability of obtaining a negative weight value from a given first transition domain.
It is easy to see that it is in fact an absorbing Markov chain with 3 transient states $\mathbf{S,M_{+},M_{-}}$ and a single absorbing state $\mathbf{T}$.
As a sidenote, the Markov chain for vanilla IS+SS, would be one without state $\mathbf{T}$ at all, terminating in either $\mathbf{M_{+}}$ or $\mathbf{M_{-}}$, whichever came first,
and characterized by the same exact $p$.
The transition matrix with the states ordered as above is:

\begin{gather}
    \begin{bmatrix}
    \mathbf{Q} & \mathbf{R} \\
    \mathbf{0} & \mathbf{I}
    \end{bmatrix} =
    \begin{bmatrix}
    0 & p & 1-p &  0 \\
    0 & p &  0  & 1-p \\
    0 & 0 & 1-p &  p \\
    0 & 0 &  0  &  1
    \end{bmatrix}
\end{gather}

\noindent From the theory of absorbing Markov chains \cite{Grinstead2003}, it is known that the average number of state transitions to absorption is given by:

\begin{gather}
    t = (\mathbf{I} - \mathbf{Q})^{-1} \cdot \mathbf{1} =
    \begin{bmatrix}
    1 & -p  & p-1 \\
    0 & 1-p &  0  \\
    0 &  0  &  p  \\
    \end{bmatrix}^{-1} \cdot \mathbf{1} =
    \begin{bmatrix}
    1 & \frac{p}{1-p} & \frac{1-p}{p} \\
    0 & \frac{1}{1-p} &      0        \\
    0 &       0       &  \frac{1}{p}   \\
    \end{bmatrix} \cdot \mathbf{1} =
    \begin{bmatrix}
    1 + \frac{p}{1-p} + \frac{1-p}{p} \\
    \frac{1}{1-p} \\
    \frac{1}{p}
    \end{bmatrix}
\end{gather}

\noindent where $\mathbf{1}$ is the vector of all ones.
The first row of the resultant vector represents the average number of transitions required to reach the absorbing state $T$ starting from $S$.
It is easy to see that the minimum of this function is attained at $p$ = 0.5, and that it is undesirable for $p$ to be very close to 0 or 1.
If $p$, and therefore $1-p$, are close to 0.5, on average 3 transitions will be required for the chain to terminate.
Lemma \ref{prop:l1} has already established that $p$ is actually 0.5 for any transition domain.

\section{Unbiasedness}
The proof of unbiasedness builds upon a useful formalism, introduced in \cite{Huang2024}, appendix A, namely that the destination electric potential of
a single random walk at each transition domain is a discrete-time martingale.
Considering the antithetic floating random walk with weight $w_i$ and destination potential $\phi_i$ as per (\ref{eq:7}),
we can enumerate the exact spatial locations it passes through as $\mathbf{r},\mathbf{r_1^{(i)}}, \mathbf{r_2^{(i)}}, ..., \mathbf{r_T^{(i)}}$ where $\phi(\mathbf{r_T^{(i)}})=\phi_i$
and similarly for its "sibling" walk with weight $\tilde{w}_i$ and destination potential $\tilde{\phi}_i$.
A result proved in \cite{Huang2024}, due to the martingale property, is that $E[\phi_{i}]=E[\phi(\mathbf{r_T^{(i)}})]=\phi(\mathbf{r_1^{(i)}})$
and also analogously that $E[\tilde{\phi}_{j}]=\phi(\mathbf{\tilde{r}_1^{(j)}})$.
Given this and taking the total expectation of (\ref{eq:8}), conditional on $\mathbf{r_1}$ and the martingales, we have:

\begin{align*}
    E[\hat{\mu}_{GAS}] & = F \frac{1}{n}(\sum\limits_{i=1}^{n}E[w_i E[\phi_{i}]] + \sum\limits_{j=1}^{n}E[\tilde{w}_j E[\tilde{\phi}_{j}]]) \\
                       & = F \frac{1}{n}(\sum\limits_{i=1}^{n}E[w_i \phi(\mathbf{r_1^{(i)}})] + \sum\limits_{j=1}^{n}E[\tilde{w}_j \phi(\mathbf{\tilde{r}_1^{(j)}})]) && \text{(martingale result)} \\
                       & = F \frac{1}{n}\sum\limits_{i=1}^{n}2E[w_i \phi(\mathbf{r_1^{(i)}})] && \text{(Lemma \ref{prop:l2} and equation (\ref{eq:theory}))} \\
                       & = 2F \cdot E[w\phi(\mathbf{r_1})] \\
                       & = F \oint_{S_1} 2\frac{K}{2L} q(\mathbf{r},\mathbf{r_1}) \phi(\mathbf{r_1}) d\mathbf{r_1} && \text{(equation (\ref{eq:7}))} \\
                       & = \oint_{S_1} F w(\mathbf{r},\mathbf{r_1}) q(\mathbf{r},\mathbf{r_1}) \phi(\mathbf{r_1}) d\mathbf{r_1} \numberthis
\end{align*}

\noindent To clarify the third step further, note that, e.g., $w_5$ vs $\tilde{w}_5$, and $\phi(\mathbf{r_1^{(5)}})$ vs $\phi(\mathbf{\tilde{r}_1^{(5)}})$,
come from the exact same first transition domain (the fifth one formed) and have the exact same marginal distributions from Lemma \ref{prop:l2}.
The last equation is exactly the inner integral of (\ref{eq:4}) and unbiasedness is now evident from the law of total expectation,
conditional now on $\mathbf{r}$ and $\mathbf{r_1}$,
i.e., the outer integral on the Gaussian surface G in (\ref{eq:4}), exactly as in \cite{Huang2024}, appendix B.

\section{Variance Reduction}

From (\ref{eq:10}), substituting the non-zero target potentials $\phi_{i}\tilde{\phi}_{i}=1$
and $\tilde{w_i} = -w_i$, covariance can be analyzed as follows:

\begin{equation}
    Cov_{GAS} = -\frac{1}{n^2}(\sum_{k=1}^{n^\prime}(E[w_k^2] - E[w_k]^2))
\end{equation}

\noindent where $n^\prime$ is the number of non-zero target potentials.
Let us examine what happens if we assume an alternative correlation-based estimator,
where surviving weight pairs are also $n^\prime$ in number and profess matching magnitudes,
and where the weights within a single pair do \textit{not} profess opposite signs.
Without loss of generality, we can assume this happens for the first weight pair, i.e.,
that $\tilde{w_1} = w_1$.
Then:

\begin{align*}
    Cov_{ALT} &= \frac{1}{n^2}(E[w_1^2] - E[w_1]^2) -\frac{1}{n^2}(\sum_{k=2}^{n^\prime}(E[w_k^2] - E[w_k]^2)) \\
              &= \frac{Var[w_1]}{n^2} -\frac{1}{n^2}(\sum_{k=2}^{n^\prime}(E[w_k^2] - E[w_k]^2)) \\
              &= \frac{2Var[w_1]}{n^2} - \frac{Var[w_1]}{n^2} - \frac{1}{n^2}(\sum_{k=2}^{n^\prime}(E[w_k^2] - E[w_k]^2)) \\
              &= \frac{2Var[w_1]}{n^2} -\frac{1}{n^2}(\sum_{k=1}^{n^\prime}(E[w_k^2] - E[w_k]^2)) \\
              &= \frac{2Var[w_1]}{n^2} + Cov_{GAS} > Cov_{GAS} \numberthis
\end{align*}

\noindent since $Var[\cdot]$ is always non-negative.
This can be extended through induction for any number of weight pairs
where the opposite sign condition is violated. Hence, in expectation,
GAS achieves maximally negative correlation among correlation-based variance reduction schemes
over any set of $n^\prime$ weight pairs with matching magnitudes.
Of course this does not, in principle, preclude a variance reduction scheme from
consistently producing $n^{\prime\prime} > n^\prime$ correlated weight pairs or weight pairs
with more favorable magnitudes for negative correlation,
however the strongly stochastic nature of the FRW algorithm
renders this an open problem and GAS a good practical option.

The above holds in expectation and a valid question could be:
Is there a chance that the actual Monte Carlo covariance from (\ref{eq:11}) turns out to be positive?
Because then the total Monte Carlo variance from (\ref{eq:9}) would increase and affect algorithm convergence.
It turns out that Monte Carlo covariance for GAS-2 is always negative and therefore always reduces Monte Carlo variance:

\begin{equation}
    \begin{split}
        \Delta &= \frac{1}{n^2n^\prime}(\sum\limits_{i=1}^{n^\prime}n^{\prime}w_{i}\tilde{w}_{i} - \sum\limits_{i=1}^{n^\prime}\sum\limits_{j=1}^{n^\prime}w_{i}\tilde{w}_{j}) \\
               &= \frac{1}{n^2n^\prime}\sum\limits_{i=1}^{n^\prime}(n^{\prime}w_{i}\tilde{w}_{i} - \sum\limits_{j=1}^{n^\prime}w_{i}\tilde{w}_{j}) \\
               &= \frac{1}{n^2n^\prime}\sum\limits_{i=1}^{n^\prime}[(n^{\prime}-1)w_{i}\tilde{w}_{i} - \sum\limits_{j\neq{i}}^{n^\prime}w_{i}\tilde{w}_{j}] \\
               &= -\frac{1}{n^2n^\prime}\sum\limits_{i=1}^{n^\prime}[(n^{\prime}-1)w_{i}^2 - \sum\limits_{j\neq{i}}^{n^\prime}w_{i}{w}_{j}] \\
               &= -\frac{1}{n^2n^\prime}[\sum\limits_{(i,j): i \leq n^\prime,j>i}(w_{i}^2 + w_{j}^2) - \sum\limits_{(i,j): i \leq n^\prime,j>i}2w_{i}{w}_{j}] \\
               &= -\frac{1}{n^2n^\prime}\sum\limits_{(i,j): i \leq n^\prime,j>i}(w_{i}^2 + w_{j}^2 - 2w_{i}{w}_{j}) \\
               &= -\frac{1}{n^2n^\prime}\sum\limits_{(i,j): i \leq n^\prime,j>i}(w_{i} - w_{j})^2 \leq 0
    \end{split}
\end{equation}

\noindent Note that intermediate results in the proof are easy to derive by expanding and using induction on $n$, e.g., for $n=3$:

\begin{equation}
    \sum\limits_{(i,j): i \leq 3,j>i}(w_{i}^2 + w_{j}^2) = (w_{1}^{2}+w_{2}^{2}) + (w_{1}^{2}+w_{3}^{2}) + (w_{2}^{2}+w_{3}^{2})
\end{equation}

To prove Theorem \ref{prop:t1}, notice that in GAS-2,
sample variance and covariance are calculated every $2n$ weight values, $n$ positive and $n$ negative,
obtained from $\mathit{n}$ first transition domains.
Assuming vanilla IS+SS has assembled approximately the same number of positive and negative weight values,
from $2n$ first transition domains, i.e., $n$ positive and $n$ negative, (\ref{eq:9})
can exactly represent the IS+SS estimator variance:

\begin{align*}
    Var[X + \tilde{X}]_{IS+SS}&= Var[X] + Var[\tilde{X}] \\
    Var[X + \tilde{X}]_{GAS}    &= Var[X + \tilde{X}]_{IS+SS} + 2\Delta, \Delta \leq 0
\end{align*}

\noindent Hence, sample variance reduction is achieved with GAS-2 in the course of an extraction compared to IS+SS.
The only remaining assumption to complete our result is that of approximate equality of positive and negative weight values.
This follows directly from Lemma \ref{prop:l1}.

\end{document}